\documentclass[a4paper, 11pt]{article}

\usepackage[utf8]{inputenc}
\usepackage[T1]{fontenc}
\usepackage{amsmath}
\usepackage{amssymb}
\usepackage{graphicx}
\usepackage{geometry}
\usepackage{amsthm} 
\usepackage{comment} 
\usepackage[dvipsnames]{xcolor}
\usepackage[inkscapelatex=false]{svg}

\usepackage{float}
\usepackage{dblfloatfix}
\usepackage[dvipsnames]{xcolor}
\usepackage{tikz}

\usetikzlibrary{shapes, arrows, positioning, fit, calc}

\usepackage{amsmath, amssymb}
\usetikzlibrary{decorations.pathreplacing, calc, positioning, patterns}

\usepackage{amsthm}

\newtheorem{theorem}{Theorem}[section] 
\newtheorem{assumption}[theorem]{Assumption} 
\newtheorem{lemma}[theorem]{Lemma}         
\newtheorem{definition}{Definition}[section] 
\newtheorem{remark}{Remark}[section]       

\usepackage{algorithm}
\usepackage{algorithmic} 
\usepackage{svg}         

\title{\LARGE \bf
Noise-Tolerant Hybrid Approach for Data-Driven Predictive Control
}

\author{Mahmood Mazare, Hossein Ramezani
\thanks{The authors are with the Department of Mechanical and Electrical Engineering, University of Southern Denmark, Denmark.
        {\tt\small \{mazare,ramezani\}@sdu.dk}}%
}
\date{} 

\begin{document}

\maketitle
\thispagestyle{empty}
\pagestyle{empty}

\begin{abstract}
This paper focuses on a key challenge in hybrid data-driven predictive control: the effect of measurement noise on Hankel matrices. While noise is handled in direct and indirect methods, hybrid approaches often overlook its impact during trajectory estimation. We propose a Noise-Tolerant Data-Driven Predictive Control (NTDPC) framework that integrates singular value decomposition to separate system dynamics from noise within reduced-order Hankel matrices. This enables accurate prediction with shorter data horizons and lower computational effort. A sensitivity index is introduced to support horizon selection under different noise levels. Simulation results indicate improved robustness and efficiency compared to existing hybrid methods.\vspace{3mm}

\end{abstract}

\section{Introduction}
Over the past decade, Data-Driven Predictive Control (DPC) has emerged as a powerful paradigm that revolutionizes control system design by eliminating the need for explicit mathematical models \cite{coulson2019regularized,verheijen2023handbook}. By directly leveraging input-output data from the physical system, DPC approaches achieve remarkable control performance while significantly reducing the engineering effort traditionally required for system identification and model development \cite{xiong2025data,krishnan2021direct}. DPC approaches can be divided into two categories: indirect and direct methods \cite{hou2013model,markovsky2021behavioral,huang2023robust}.

Indirect methods, exemplified by Subspace Predictive Control (SPC) \cite{favoreel1999spc}, first identify a prediction model from data, which is subsequently used as a constraint in the optimal control problem. This intermediate modeling phase makes these approaches particularly suitable for integration with large datasets. Such integration allows for reduction in predictor bias without increasing the computational complexity of the design problem \cite{berberich2024overview}.

In contrast, direct DPC approaches, such as Data-Enabled Predictive Control (DeePC) \cite{coulson2019data}, instead of estimating a prediction model incorporate the input output samples directly in the optimal control problem \cite{zhang2023dimension}. This is possible by introducing a regularization vector "$\mathbf{g}(k)$" based on Willems' fundamental lemma \cite{willems2005note}. This approach enables the predictor to be "control-oriented" specifically focused on generating predictions for control purposes, however, the optimization problem has a larger size as $\mathbf{g}(k)$ is part of the decision variables that needs to be optimized \cite{alsalti2024sample}.

Direct DPC methods require regularization of $\mathbf{g}(k)$, whose dimension depends on data window length $M$. Since $M$ must be significantly larger than the prediction/control horizons to ensure sufficient excitation, the optimization problem becomes computationally intensive and tuning-sensitive. To address these challenges, hybrid methods are proposed to pre-compute $\mathbf{g}(k)$ before formulating the main control problem \cite{fiedler2021relationship}. These approaches typically employ matrix factorization techniques to select an appropriate solution for $\mathbf{g}(k)$ from the non-unique set of possibilities established by the fundamental lemma \cite{de2025kernelized}.

As two examples of hybrid methods, \cite{yang2015data,smith2024optimal} decompose $\mathbf{g}(k)$ into the past ($\mathbf{g}_\text{ini}(k)$) and the future ($\mathbf{g}_{f}(k)$) components calculated from the past and the future trajectories, respectively. In \cite{yang2015data}, $\mathbf{g}_\text{ini}(k)$ is calculated through a LS problem on the past data while $\mathbf{g}_{f}$ is the solution of an optimization problem on the future data. In \cite{smith2024optimal} however, $\mathbf{g}_\text{ini}(k)$ and $\mathbf{g}_{f}(k)$ are considered to be orthogonal. This assumption restricts the search space for $\mathbf{g}_{f}(k)$ as it should belong to the null space of the past Hankel matrix. After calculating $\mathbf{g}(k)$, the remaining process of calculating the control signal in hybrid methods is more similar to the indirect methods. Compared to the SPC method as the main representative of indirect methods, both hybrid approaches are computationally more efficient as they work with two lower order Hankel matrices rather than a big matrix. However, when handling noisy data their efficiency drops as the estimation of $\mathbf{g}(k)$ components does not consider noisy measurement. This is even more important when knowing that SPC as a base of comparison is efficient for noisy cases as well.

To improve the noise handling of hybrid approaches, this study proposes a different approach for estimation of $\mathbf{g}_\text{ini}(k)$ and $\mathbf{g}_f(k)$ that takes into account the noise characteristics. Compared to the mentioned hybrid methods, our approach borrows the assumption of orthogonality of components from \cite{smith2024optimal} and replaces the LQ factorization with SVD which can exctract the system dynamics from the noisy data when estimating the $\mathbf{g}(k)$ components. 
\textcolor{black}{The application of Singular Value Decomposition (SVD) in data-driven control has recently gained attention, primarily for computational model order reduction within \textit{direct} control formulations \cite{alsalti2024robust,zhang2023dimension, shi2024efficient}. However, the potential of SVD to geometrically separate signal and noise subspaces within a \textit{hybrid} predictive framework remains underexplored. This study uses SVD beyond simple dimensionality reduction, utilizing it as a robust filtering mechanism. By embedding this projection into the predictor construction, the proposed NTDPC framework systematically replaces the heuristic regularization typical of DeePC with a rigorous, noise-tolerant estimation step.}
 In \cite{alsalti2024robust}, using SVD the length of past data and hence the dimension of Hankel matrix is reduced. \cite{zhang2023dimension,shi2024efficient} propose different applications for SVD in which the regularization vector g can be replaced with a lower dimension vector in direct methods. However, when it comes to hybrid methods, using SVD approach is new and even more important. The fact that factorization is used to select one solution among infinite possible solutions for $\mathbf{g}(k)$ components, makes it crucial to select the solution with more information from the system dynamics rather than the noise. The procedure, importance and efficiency of using SVD in hybrid structure are the highlights of this study. \par
In general the number of columns in Hankel matrices is significantly higher than the number of rows to fulfill the persistency of excitation property. On the other hand the number of rows is related to the prediction horizon which in turn is usually much higher than the system order. This implies that (according to the fundamental lemma) the rank of Hankel matrix is even less than the number of rows hence the matrix is rank deficit. However, when the output samples in a Hankel matrix are filled with noisy measurements, the Hankel matrix becomes full row rank. This makes the process of selecting optimum $\mathbf{g}_\text{ini}(k)$ and $\mathbf{g}_f(k)$ more important as these optimum vectors should represent the system dynamics rather than the noise. Regarding the noise handling in existing hybrid methods, \cite{yang2015data} has only considered noise free samples and in \cite{smith2024optimal} the noise is only present in the signals not in the Hankel matrices. This makes the application of those methods limited as in practice there is no access to noise free output signals. Assuming noise for both signals and Hankel matrices, this paper proposes a systematic approach to separate system dynamics from noise using SVD factorization. Thanks to the nature of this factorization and given the system order, Hankel matrices will be decomposed into two parts associated with system dynamics and noise, respectively. \par 
The primary idea behind factorization in both \cite{smith2024optimal} and our study is to replace the regularization vector $\mathbf{g}(k)$ with a lower order vector since the Hankel matrix is not full rank. However, our work differs from \cite{smith2024optimal} when estimating this lower order regularization vector. Since the measurement noise only affects the output not input, estimation is in the form of an optimization problem with a mixture of deterministic and stochastic equations. When using LQ factorization as in \cite{smith2024optimal}, the deterministic part can be solved separately, thanks to the triangular matrix in the equation. This is not the case in our approach where the deterministic equation is not separately solvable. Therefore the overall optimization problem should consider both deterministic and stochastic equations. Regarding the validation of the proposed algorithm, the simulation results are compared with the hybrid method of \cite{smith2024optimal} as well as the optimum solution of SPC method as a reference. The comparison with SPC is also made from the computational effort point of view. Moreover, a sensitivity index is introduced that can help for the selection of predictive control horizon. \vspace{3mm}

\textbf{Notations:} Throughout this paper, the sets of natural numbers and real numbers are denoted by \(\mathbb{N}\) and \(\mathbb{R}\), respectively. The set of real vectors with dimensions \(n \times 1\) is denoted as \(\mathbb{R}^n\). We denote the set of natural numbers excluding zero as \(\mathbb{N}_{\geq 1}\). For any finite collection of \(\{ \xi_1, \ldots, \xi_q \} \in \mathbb{R}^{n_1} \times \cdots \times \mathbb{R}^{n_q}\), where \(q \in \mathbb{N}_{\geq 1}\), we use the operator \(\text{col}(\xi_1, \ldots, \xi_q) := [\xi_1^\top, \ldots, \xi_q^\top]^\top\) to denote the vector formed by stacking these vectors in a column. Given a matrix \(A \in \mathbb{R}^{n \times m}\), its Frobenius norm is denoted by \(\|A\|_{\text{Fro}}\), and its generalized weighted pseudo-inverse is denoted by \(A_W^\dagger:=(A^T W A)^{-1}A^T W\).

\section{PRELIMINARIES}

Consider a discrete-time linear dynamical system influenced by zero-mean Gaussian noise $\zeta(k) \sim \mathcal{N}(0, \Sigma_\zeta)$:
\begin{align}
    x(k + 1) &= Ax(k) + Bu(k), \quad k \in \mathbb{N}, \\
    y(k) &= Cx(k) + \zeta(k),
\end{align}
where $x \in \mathbb{R}^n$ is the state, $u \in \mathbb{R}^{n_u}$ is the control input, $y \in \mathbb{R}^{n_y}$ is the measured output, and $A, B, C$ are real matrices of appropriate dimensions. We assume that $(A, B)$ is controllable and $(A, C)$ is observable. By applying a persistently exciting input sequence $\{u(k)\}_{k=0}^{T-1}$ of length $T$ to system (1), we obtain a corresponding output sequence $\{y(k)\}_{k=0}^{T-1}$. The measurement noise $\zeta(k)$ is assumed to be i.i.d., zero-mean Gaussian, and uncorrelated in time. \\
\textcolor{black}{For an input-output model, we introduce the past horizon $T_{\text{ini}}$ to characterize the system's memory.
Let $y_o(k)$ denote the underlying \textit{noise-free} output of the system (i.e., the output of (1)-(2) if $\zeta(k) \equiv 0$). 
The system dynamics and the measurement model are described by:
\begin{subequations}
\begin{align} 
    y_o(k) &= \sum_{i=1}^{T_{\text{ini}}} a_i y_o(k - i) + \sum_{i=1}^{T_{\text{ini}}} b_i u(k - i), \label{eq:diff_eq_noiseless} \\
    y(k) &= y_o(k) + \zeta(k). \label{eq:measurement_noise}
\end{align}
\end{subequations}
In this representation, the coefficients \( a_i \) and \( b_i \) define the system dynamics. For a general MIMO system, these are real-valued matrices with dimensions \( a_i \in \mathbb{R}^{n_y \times n_y} \) and \( b_i \in \mathbb{R}^{n_y \times n_u} \), while for a SISO system, they reduce to scalar parameters. Furthermore, for simplicity of notation, we assume the same memory horizon \( T_{\text{ini}} \) for both the input and output regressors, rather than assigning distinct lag lengths to \( u \) and \( y \).
}

Now, for any starting time instant $k \geq 0$ within the available data vector and for any data vector length $j \ge 1$, we define the trajectory vectors:
\begin{align}
    \bar{u}(k, j) &:= \operatorname{col}(u(k), \dots, u(k+j-1)), \nonumber \\
    \bar{y}(k, j) &:= \operatorname{col}(y(k), \dots, y(k+j-1)).
\end{align}
\textcolor{black}{In this notation, the first argument $k$ denotes the starting time instant, and the second argument $j$ explicitly denotes the length of the data window.}

\textcolor{black}{Let us define $N$ as the future prediction horizon. Using $T_{\text{ini}}$, we construct the following sequences of past and future inputs and outputs:
\begin{equation}
    \begin{split}
        \mathbf{u}_{\text{ini}}(k) &:= \text{col}(u(k-T_{\text{ini}}), \ldots, u(k-1)), \\
        \mathbf{y}_{\text{ini}}(k) &:= \text{col}(y(k-T_{\text{ini}}+1), \ldots, y(k)),\\
        \mathbf{u}_N(k) &:= \text{col}(u(k), \ldots, u(k+N-1)),\\
        \mathbf{y}_N(k) &:= \text{col}(y(k), \ldots, y(k+N-1)).
    \end{split}
    \label{eq:traj_vectors}
\end{equation}
}

To transition from model-based prediction to a data-driven counterpart, we assume access to an informative set of input-output pairs $\{\bar{u}, \bar{y}\}$ of length $T$. Informativity is ensured by the persistent excitation of the input sequence, where the rank of $\mathbf{Z}_p$ matrix defined in (\ref{EqZp}) must be higher than $T_{\text{ini}} + N + n$. \textcolor{black}{This implies that the data length $T$ must be sufficiently long, specifically $T \geq (n_u + 1)(T_{\text{ini}} + N + n) - 1$.
The persistency of excitation condition is imposed on the input sequence $\{\bar{u}\}$ , which guarantees that the corresponding Hankel matrix $\mathbf{Z}_p$ constructed from the input-output data has sufficient rank.}

\textcolor{black}{Let $M = T - T_{ini} - N + 1$ be the number of columns in the Hankel matrices. The matrices are defined as:}
\begin{equation} \label{eq:hankel}
\textcolor{black}{
\begin{aligned}
\mathbf{U}_p &:= [\bar{u}(0, T_{ini}), \dots, \bar{u}(M-1, T_{ini})], \\
\mathbf{Y}_p &:= [\bar{y}(0, T_{ini}), \dots, \bar{y}(M-1, T_{ini})], \\
\mathbf{U}_f  &:= [\bar{u}(T_{ini}, N), \dots, \bar{u}(T_{ini}+M-1, N)], \\
\mathbf{Y}_f &:= [\bar{y}(T_{ini}, N), \dots, \bar{y}(T_{ini}+M-1, N)].
\end{aligned}
}
\end{equation}
These matrices form the foundation of the data-driven predictive control approach. Next, we introduce three different DPC strategies that utilize this data directly, avoiding the need for an explicit system model or state estimator.

\vspace{3mm}
\textbf{Problem 1 [SPC]}: Subspace Predictive Control (SPC) represents an early indirect data-driven MPC technique. It directly incorporates results from subspace identification into the MPC design by estimating prediction matrices from input/output data based on the relationship $\mathbf{y}_N(k) = \begin{bmatrix} \mathbf{P}_1 & \mathbf{P}_2 & \mathbf{\Gamma} \end{bmatrix} \begin{bmatrix} \mathbf{u}_{\text{ini}}(k)^T & \mathbf{y}_{\text{ini}}(k)^T & \mathbf{u}_N(k)^T \end{bmatrix}^T$. The unknown matrix $\mathbf{\Theta} = \begin{bmatrix} \mathbf{P}_1 & \mathbf{P}_2 & \mathbf{\Gamma} \end{bmatrix}$ is estimated by solving the least-squares problem $ \min_{\mathbf{\Theta}} \left\|  \mathbf{Y}_f - \mathbf{\Theta} \begin{bmatrix} \mathbf{U}_p^T & \mathbf{Y}_p^T & \mathbf{U}_f^T \end{bmatrix}^T \right\|_{\text{Fro}}^2$, yielding the solution $\mathbf{\Theta}^* = \begin{bmatrix} \mathbf{P}_1 & \mathbf{P}_2 & \mathbf{\hat{\Gamma}} \end{bmatrix}$. The SPC problem is then formulated as \cite{favoreel1999spc}:
\begin{equation} \label{eq:SPC_problem}
    \begin{split}
        \min_{\mathbf{u}_N(k), \mathbf{y}_N(k)} \quad & \| Q^{\frac{1}{2}}(\mathbf{y}_N(k) - r_{y}) \|_2^2 + \| R^{\frac{1}{2}}(\mathbf{u}_N(k) - r_{u}) \|_2^2 \\
        \text{subject to} \quad & \mathbf{y}_N(k) = \mathbf{P}_1 \mathbf{u}_{\text{ini}}(k) + \mathbf{P}_2 \mathbf{y}_{\text{ini}}(k) + \mathbf{\hat{\Gamma}} \mathbf{u}_N(k), \\
                                & \mathbf{u}_N(k) \in \mathcal{U}, \quad \mathbf{y}_N(k) \in \mathcal{Y}.
    \end{split}
\end{equation}

\textcolor{black}{Unlike SPC, which relies on an intermediate identification step to find prediction matrices, DeePC is a direct approach that embeds the raw data samples into the optimization problem.}

\vspace{3mm}
\textbf{Problem 2 [DeePC]}:
Given immediate past input and measured output signals, find a future input sequence by solving \cite{dorfler2022bridging}:
\begin{equation}
    \begin{split}
        \min_{\Theta(k)} \quad & \| Q^{\frac{1}{2}}(\mathbf{y}_N(k) - r_{y}) \|_2^2 + \| R^{\frac{1}{2}}(\mathbf{u}_N(k) - r_{u}) \|_2^2 + \lambda_g l_g(\mathbf{g}(k)) + \lambda_y \|\sigma_y(k)\|_2^2 \\
        \text{subject to} \quad &  \begin{bmatrix}
            \mathbf{U}_p \\
            \mathbf{Y}_p \\
            \mathbf{U}_f \\
            \mathbf{Y}_f
        \end{bmatrix}
        \mathbf{g}(k) =
        \begin{bmatrix}
            \mathbf{u}_{\text{ini}}(k) \\
            \mathbf{y}_{\text{ini}}(k)+\sigma_y(k) \\
            \mathbf{u}_N(k) \\
            \mathbf{y}_N(k)
        \end{bmatrix}, \\
        & \mathbf{u}_N(k) \in \mathcal{U}, \quad \mathbf{y}_N(k) \in \mathcal{Y}
    \end{split}
\end{equation}
where $\Theta(k) =\{\mathbf{u}_N(k), \mathbf{y}_N(k), \mathbf{g}(k), \sigma_y(k)\}$ and $l_g(\mathbf{g}(k))$ is a user-defined regularization function over $\mathbf{g}(k)$, often chosen as $\|\mathbf{g}(k)\|_2^2$. The variable $\sigma_y(k)$ is a slack variable introduced to ensure feasibility in the presence of noise.

\textcolor{black}{While DeePC offers a direct data-driven formulation, the optimization vector $\mathbf{g}(k)$ scales with the dataset size, which can be computationally intensive. To mitigate this and improve noise handling, hybrid approaches have been proposed that pre-compute an optimal predictor offline.}

\vspace{3mm}
\textbf{Problem 3 [SMMPC]}: A variant of DPC, Signal Matrix Model Predictive Control (SMMPC) \cite{smith2024optimal}, employs a minimum-variance unbiased predictor derived offline. The control problem is solved as:
\begin{equation} \label{eq:SMMPC_constraint}
    \begin{split}
        \min_{\mathbf{u}_N(k), \mathbf{y}_N(k)} \quad & \| Q^{\frac{1}{2}}(\mathbf{y}_N(k) - r_{y}) \|_2^2 + \| R^{\frac{1}{2}}(\mathbf{u}_N(k) - r_{u}) \|_2^2 \\
        \text{subject to} \quad & \mathbf{y}_N(k) = \mathbf{E}_1 \mathbf{u}_{\text{ini}}(k) + \mathbf{E}_2 \mathbf{y}_{\text{ini}}(k) + \mathbf{E}_3 \mathbf{u}_N(k),  \\
                                & \mathbf{u}_N(k) \in \mathcal{U}, \quad \mathbf{y}_N(k) \in \mathcal{Y}.
    \end{split}
\end{equation}
\textcolor{black}{The computational efficiency of SMMPC arises from the offline calculation of the predictor matrices $\mathbf{E}_1$, $\mathbf{E}_2$, and $\mathbf{E}_3$, which replaces the high-dimensional online optimization of $\mathbf{g}(k)$ required in DeePC.}

\section{Indirect DPC using reduced-order Hankel matrices}\label{sec:indirect_dpc_reduced_hankel}

In data-driven predictive control, the Hankel matrices constructed from historical data typically possess significantly more columns than rows ($M \gg (n_u+n_y)T_{ini}$). This structural characteristic renders the linear system for the regularization vector $\mathbf{g}(k)$ underdetermined, admitting an infinite set of possible solutions. Consequently, standard approaches require regularization terms (e.g., minimum norm) to select a specific solution.\\
To address the underdetermined nature of the problem, we adopt a sequential identification strategy that decomposes the regularization vector $\mathbf{g}(k)$ into two orthogonal components: a past component $\mathbf{g}_{\text{ini}}(k)$ determined by historical data, and a future component $\mathbf{g}_{\text{f}}(k)$ that governs the predictive behavior. This decomposition is formulated as:
\begin{equation}
    \begin{bmatrix}
        \mathbf{U}_p \\
        \mathbf{Y}_p \\
        \mathbf{U}_f \\
        \mathbf{Y}_f
    \end{bmatrix}
    \big(\mathbf{g}_{\text{ini}}(k)+\mathbf{g}_{\text{f}}(k)\big) =
    \begin{bmatrix}
        \mathbf{u}_{\text{ini}}(k) \\
        \mathbf{y}_{\text{ini}}(k) \\
        \mathbf{u}_N(k) \\
        \mathbf{y}_N(k)
    \end{bmatrix}.
    \label{eq:UY_g_uy}
\end{equation}
\textcolor{black}{Equation \eqref{eq:UY_g_uy} poses an identification challenge due to the rank properties of the data matrix. In the noise-free case, the Hankel matrix is rank-deficient (rank equal to the system order plus the input horizon), reflecting the underlying low-dimensional dynamics. However, measurement noise invariably renders the matrix full row rank. In this full-rank scenario, standard approaches like LQ factorization force a unique solution that effectively fits the noise, treating stochastic perturbations as valid system dynamics. To prevent this overfitting, we propose a transformation that enforces the underlying low-rank structure. By assuming a spectral gap, we utilize SVD to separate the dominant signal subspace ($\Sigma_1$) from the noise subspace ($\Sigma_2$). This allows us to project the high-dimensional $\mathbf{g}(k)$ onto a lower-dimensional latent basis ($\eta_1, \eta_2$), converting the ill-posed noisy problem into a well-posed, noise-tolerant estimation problem.}\\

\subsection{Derivation of a Parsimonious Formulation}

We condense the mapping in \eqref{eq:UY_g_uy} into a minimal parameterization. Consider the stacked Hankel matrices representing the past and future input-output sequences:
\begin{equation}\label{EqZp}
    \mathbf{Z}_p := \begin{bmatrix} \mathbf{U}_p \\ \mathbf{Y}_p \end{bmatrix}, \quad
    \mathbf{Z}_f := \begin{bmatrix} \mathbf{U}_f \\ \mathbf{Y}_f \end{bmatrix}, \quad
    \mathbf{z}_{\text{ini}}(k) := \begin{bmatrix} \mathbf{u}_{\text{ini}}(k) \\ \mathbf{y}_{\text{ini}}(k) \end{bmatrix}.
\end{equation}

\begin{figure}[t]
    \centering
    \resizebox{\columnwidth}{!}{
    \begin{tikzpicture}[every node/.style={transform shape}]

        \def\H{3.0}      
        \def\Wp{5.0}     
        \def\Ws{5.0}     
        \def\R{1.2}      
        \def\Wsize{3.0}  
        
        \def\gap{0.5}          
        \def\bracegap{0.8}     
        
        \draw[thick, fill=green!10] (0,0) rectangle (\Wp, \H);
        \node at (0.5*\Wp, 0.5*\H) {\Huge $Z_p$};
        \node[anchor=south] at (0.5*\Wp, 0.1) {\large Past Data};
        
        \node at (\Wp + \gap, 0.5*\H) {\Huge $=$};
        
        \begin{scope}[shift={(\Wp + 2*\gap, 0)}]
            
            \draw[thick, fill=yellow!10] (0,0) rectangle (\Wsize, \H);
            
            \node[black!10] at (0.5*\Wsize, 0.5*\H) {\scalebox{5}{$W$}};
            
            \draw[dashed, thick] (\R, 0) -- (\R, \H);
            
            \node at (0.25*\R, 0.5*\H) [rotate=90] {\large -- $w_1$ --};
            \node at (0.6*\R, 0.5*\H) {$\dots$};
            \node at (0.85*\R, 0.5*\H) [rotate=90] {\large -- $w_r$ --};
            
            \node at (\R + 0.35, 0.5*\H) [rotate=90] {\large -- $w_{r+1}$ --};
            \node at (\R + 0.5*\Wsize - 0.5*\R, 0.5*\H) {$\dots$};
            \node at (\Wsize - 0.5, 0.5*\H) [rotate=90] {\large -- $w_M$ --};
            
            \draw [decorate,decoration={brace,amplitude=6pt}, thick]
                (0,\H+0.2) -- (\R,\H+0.2) node [midway,above,yshift=6pt] {\Large $W_1$};
                
            \draw [decorate,decoration={brace,amplitude=6pt}, thick]
                (\R,\H+0.2) -- (\Wsize,\H+0.2) node [midway,above,yshift=6pt] {\Large $W_2$};
        \end{scope}
        
        \begin{scope}[shift={(\Wp + 2*\gap + \Wsize + \bracegap, 0)}]
            \draw[thick] (0,0) rectangle (\Ws, \H);
            
            \draw[fill=black!5] (0, \H) rectangle (\R, \H-\R);
            
            \node[black!15] at (0.5*\R, \H-0.5*\R) {\scalebox{2.0}{$\Sigma_1$}};
            
            \node[black!90!black, font=\bfseries] at (0.2, \H-0.25) {$\sigma_1$};
            \node[black!90!black, font=\bfseries] at (\R-0.3, \H-\R+0.25) {$\sigma_r$};
            \draw[dotted, very thick, black] (0.5, \H-0.5) -- (\R-0.6, \H-\R+0.6);

            \draw[fill=red!10, dashed] (\R, \H-\R) rectangle (\Ws, 0);
            \node[red!60!black] at (\R + 0.5*\Ws - 0.5*\R, 0.5*\H - 0.5*\R) {\huge $\Sigma_2 \ll $};
            
            \node at (\Ws-0.5, \H-0.5) {\Large $0$};
            \node at (0.5, 0.5) {\Large $0$};
            
            \draw [decorate,decoration={brace,amplitude=6pt}, thick]
                (0,\H+0.2) -- (\R,\H+0.2) node [midway,above,yshift=6pt] {\large Signal};
            \draw [decorate,decoration={brace,amplitude=6pt}, thick]
                (\R,\H+0.2) -- (\Ws,\H+0.2) node [midway,above,yshift=6pt] {\large Noise};
        \end{scope}
        
        
        \begin{scope}[shift={(\Wp + 2*\gap + \Wsize + \bracegap + \Ws + 2*\gap, {0.5*\H - 0.5*\Ws})}]
            
            \draw[thick, fill=orange!10] (0,0) rectangle (\Ws, \Ws);
            \node[black!10] at (0.5*\Ws, 0.5*\Ws) {\scalebox{5}{$V^\top$}};
            
            \draw[thick, fill=orange!20] (0, \Ws) rectangle (\Ws, \Ws-\R);
            \node at (0.5*\Ws, \Ws-0.3) {--- $v_1^\top$ ---};
            \node at (0.5*\Ws, \Ws-0.65) {$\vdots$};
            \node at (0.5*\Ws, \Ws-\R+0.3) {--- $v_r^\top$ ---};
            
            \node at (0.5*\Ws, \Ws-\R-0.4) {--- $v_{r+1}^\top$ ---};
            \node at (0.5*\Ws, 0.5*\Ws - 0.5*\R) {$\vdots$};
            \node at (0.5*\Ws, 0.4) {--- $v_M^\top$ ---};

            \draw[dashed, thick] (0, \Ws-\R) -- (\Ws, \Ws-\R);

            \draw [decorate,decoration={brace,amplitude=6pt}, thick]
                (\Ws+0.2, \Ws) -- (\Ws+0.2, \Ws-\R) node [midway,right,xshift=6pt] {\Large $V_1^\top$ (Range)};
                
            \draw [decorate,decoration={brace,amplitude=6pt}, thick]
                (\Ws+0.2, \Ws-\R) -- (\Ws+0.2, 0) node [midway,right,xshift=6pt] {\Large $V_2^\top$ (Null)};
        \end{scope}

    \end{tikzpicture}
    }
    \caption{\textcolor{black}{SVD Decomposition of the noisy Hankel matrix $Z_p$. $\Sigma_1$ captures dominant dynamics; $\Sigma_2$ represents noise. $W$ is partitioned into signal columns $W_1$ and noise columns $W_2$.}}
    \label{fig:svd_decomp}
\end{figure}
By performing SVD factorization on the past data matrix $\mathbf{Z}_p$:
\begin{equation}
    \mathbf{Z}_p  = W \Sigma V^T.
\end{equation}
where $W$ and $V$ are unitary matrices. In the ideal noise-free case ($\zeta=0$), the matrix $\Sigma$ is rank-deficient. \textcolor{black}{We can partition the decomposition as:}
\begin{equation}\label{eq:Zp_SVD_L1V1T}
    \mathbf{Z}_p  = \begin{bmatrix} W_1 & W_2 \end{bmatrix} \begin{bmatrix} \Sigma_1 & 0 \\ 0 & 0 \end{bmatrix} \begin{bmatrix} V_1^T \\ V_2^T \end{bmatrix} = W_1 \Sigma_1 V_1^T = L_1 V_1^T,
\end{equation}
\textcolor{black}{where we define $L_1 := W_1 \Sigma_1$. Here, $\Sigma_1$ is a diagonal matrix containing the non-zero singular values (see Fig.\ref{fig:svd_decomp}). The rank of the noiseless $\mathbf{Z}_p$ is given by $r = n_u T_{\text{ini}} + n$ \cite{van1996subspace}. Consequently, the dimensions are defined as $\Sigma_1 \in \mathbb{R}^{r \times r}$ and $V_1^T \in \mathbb{R}^{r \times M}$, where $M$ is the number of columns in $\mathbf{Z}_p$.}

In cases of low Signal-to-Noise Ratio (SNR), the singular values of the system dynamics and noise can indeed overlap or "flip," making simple truncation problematic.
Therefore, we explicitly clarify that our method operates under the assumption of a sufficient spectral gap between the dominant system dynamics and the noise floor.

\textcolor{black}{Substituting this decomposition into the past trajectory equation \eqref{eq:UY_g_uy}:}
\begin{equation}
    \mathbf{z}_{\text{ini}}(k) = \mathbf{Z}_p \big(\mathbf{g}_{\text{ini}}(k) + \mathbf{g}_{\text{f}}(k)\big).
    \label{eq:z_ini_Zp_g_sum}
\end{equation}
\textcolor{black}{We enforce the condition that $\mathbf{g}_{\text{f}}(k)$ lies in the null space of $\mathbf{Z}_p$, such that $\mathbf{Z}_p \mathbf{g}_{\text{f}}(k) = 0$. This simplifies \eqref{eq:z_ini_Zp_g_sum} to:}
\begin{equation}
    \mathbf{z}_{\text{ini}}(k) = \mathbf{Z}_p \mathbf{g}_{\text{ini}}(k) = L_1 V_1^T \mathbf{g}_{\text{ini}}(k).
    \label{eq:z_ini_L1V1T_gini}
\end{equation}
\textcolor{black}{We now define the latent vector $\eta_1(k) \in \mathbb{R}^{r}$ as the projection of the regularization vector onto the active subspace:}
\begin{equation}
    \eta_1(k) = V_1^T \mathbf{g}_{\text{ini}}(k).
\end{equation}
\textcolor{black}{This yields the following linear system describing the initial trajectory:}
\begin{equation}
    \mathbf{z}_{\text{ini}}(k) = L_1 \eta_1(k).
    \label{eq:z_ini_eta}
\end{equation}
\textcolor{black}{Since we seek the minimum-norm solution for $\mathbf{g}_{\text{ini}}(k)$ that satisfies the data consistency, we restrict $\mathbf{g}_{\text{ini}}(k)$ to the row space of $\mathbf{Z}_p$ (spanned by $V_1$). This allows us to recover $\mathbf{g}_{\text{ini}}(k)$ uniquely from $\eta_1(k)$ via:}
\begin{equation}
    \mathbf{g}_{\text{ini}}(k) = V_1 \eta_1(k).
    \label{eq:g_ini_eta2}
\end{equation}

\textcolor{black}{In the realistic case where $\zeta \neq 0$, the measured output matrix $\mathbf{Y}_p$ is noisy, and $\mathbf{Z}_p$ becomes full rank. Let $M$ be the number of columns. By utilizing  SVD, yielding a rectangular singular value matrix $\Sigma \in \mathbb{R}^{(n_u + n_y)T_{ini} \times M}$, which is partitioned as:
\begin{equation}\label{eq:ss}
\Sigma = \begin{bmatrix} \Sigma_1 & 0 \\ 0 & \Sigma_2 \end{bmatrix},
\end{equation}
where $\Sigma_1 \in \mathbb{R}^{r \times r}$ contains the $r$ dominant singular values associated with the system dynamics. The second block, $\Sigma_2 \in \mathbb{R}^{((n_u + n_y)T_{ini}-r) \times (M-r)}$, is a rectangular matrix containing the remaining singular values dominated by noise (padded with zeros as necessary). By truncating $\Sigma_2$, we obtain a low-rank approximation that serves as a robust basis for estimating $\eta_1(k)$.}

\subsection{Best Linear Unbiased Estimator for \(\eta_1(k)\)}
\label{subsec:black_eta1}

We now formalize the estimation of the latent vector \(\eta_1(k)\) in the presence of noise. \textcolor{black}{This formulation distinguishes between the deterministic inputs (controller outputs) and the stochastic measurements (noisy sensor measurements) using the Best Linear Unbiased Estimator (black) framework.
In the following derivation, we treat the matrix $L_y$ as a deterministic basis for the system dynamics. This is justified by the offline SVD truncation step (Section 3.1), which uses the large data history $M$ to filter out measurement noise and extract the dominant signal subspace. Consequently, while the offline raw data is noisy, the resulting basis $L_y$ represents the denoised nominal model, whereas the single online measurement $z_{ini,m}(k)$ retains significant noise that must be filtered.}

\begin{theorem} \label{thm:black_eta1}\textcolor{black}{
Let the measured initial trajectory be \(\mathbf{z}_{\text{ini},m}(k) = \begin{bmatrix} \mathbf{u}_{\text{ini}}(k)^\top & \mathbf{y}_{\text{ini},m}(k)^\top \end{bmatrix}^\top\), where \(\mathbf{y}_{\text{ini},m}(k) = \mathbf{y}_{\text{ini}}(k) + \boldsymbol{\zeta}_{\text{ini}}(k)\), with measurement noise \(\boldsymbol{\zeta}_{\text{ini}}(k) \sim \mathcal{N}(0, \mathbf{\Sigma}_{\zeta})\), and $\mathbf{z}_{\text{ini},m}(k) = L_1 \eta_1(k)$.
By partitioning the system basis matrix \(L_1\) from \eqref{eq:Zp_SVD_L1V1T} as \(L_1 = \begin{bmatrix} L_u^\top & L_y^\top \end{bmatrix}^\top\), where \(L_u \in \mathbb{R}^{n_u T_{\text{ini}} \times r}\) and \(L_y \in \mathbb{R}^{n_y T_{\text{ini}} \times r}\),
the Best Linear Unbiased Estimator (black) of \(\eta_1(k)\) is given by:
\begin{equation}\label{eq:black_eta1}
    \hat{\eta}_1(k) = \mathcal{K} \mathbf{z}_{\text{ini},m}(k),
\end{equation}
where the estimator gain matrix \(\mathcal{K} \in \mathbb{R}^{r \times (n_u + n_y)T_{\text{ini}}}\) is defined as:
\begin{equation}\label{eq:K_gain_matrix}
    \mathcal{K} = \begin{bmatrix}
        L_u^{\dagger} - K_y L_y L_u^{\dagger} & K_y
    \end{bmatrix},
\end{equation}
with the auxiliary correction gain \(K_y\) given by:
\begin{equation}\label{eq:Ky_def}
    K_y = P \left( P^\top L_y^\top \mathbf{\Sigma}_{\zeta}^{-1} L_y P \right)^{\dagger} P^\top L_y^\top \mathbf{\Sigma}_{\zeta}^{-1},
\end{equation}
and \(P = I - L_u^{\dagger} L_u\) is the orthogonal projector onto the null space of \(L_u\).}
\end{theorem}

\begin{proof}\textcolor{black}{
The estimation problem is formulated as a constrained generalized least-squares optimization. We seek the estimate \(\hat{\eta}_1(k)\) that minimizes the variance of the measurement residual subject to the deterministic input constraints:
\begin{equation}\label{eq:ls_problem}
    \begin{aligned}
        \min_{\eta_1} \quad & J = \left( \mathbf{y}_{\text{ini},m}(k) - L_y \eta_1 \right)^\top \mathbf{\Sigma}_{\zeta}^{-1} \left( \mathbf{y}_{\text{ini},m}(k) - L_y \eta_1 \right) \\
        \text{s.t.} \quad & L_u \eta_1 = \mathbf{u}_{\text{ini}}(k).
    \end{aligned}
\end{equation}
The general solution satisfying the linear constraint is \(\eta_1 = L_u^{\dagger} \mathbf{u}_{\text{ini}}(k) + P \xi\), where \(\xi\) is an arbitrary vector in the null space of \(L_u\). Substituting this parameterization into the objective function \(J\) transforms the problem into an unconstrained minimization over \(\xi\):
\begin{equation}
    \min_{\xi} \left\| \mathbf{\Sigma}_{\zeta}^{-1/2} \left( \tilde{\mathbf{y}}(k) - L_y P \xi \right) \right\|_2^2,
\end{equation}
where \(\tilde{\mathbf{y}}(k) = \mathbf{y}_{\text{ini},m}(k) - L_y L_u^{\dagger} \mathbf{u}_{\text{ini}}(k)\) represents the output residual after removing the effect of the known inputs. The optimal \(\xi\) is obtained via the weighted pseudoinverse:
\begin{equation}
    \hat{\xi} = \left( L_y P \right)^{\dagger} \mathbf{\Sigma}_{\zeta}^{-1}\tilde{\mathbf{y}}(k) = \left( P^\top L_y^\top \mathbf{\Sigma}_{\zeta}^{-1} L_y P \right)^{\dagger} P^\top L_y^\top \mathbf{\Sigma}_{\zeta}^{-1} \tilde{\mathbf{y}}(k).
\end{equation}
Substituting \(\hat{\xi}\) back into the expression for \(\eta_1\) yields:
\begin{equation}
    \hat{\eta}_1(k) = L_u^{\dagger} \mathbf{u}_{\text{ini}}(k) + P \hat{\xi} = L_u^{\dagger} \mathbf{u}_{\text{ini}}(k) + K_y \left( \mathbf{y}_{\text{ini},m}(k) - L_y L_u^{\dagger} \mathbf{u}_{\text{ini}}(k) \right).
\end{equation}
Grouping the terms multiplying \(\mathbf{u}_{\text{ini}}(k)\) and \(\mathbf{y}_{\text{ini},m}(k)\) results in the block matrix form presented in \eqref{eq:K_gain_matrix}. We assume that the output subspace \(L_y P\) has sufficient rank to identify the free components of \(\eta_1\).
Crucially, this estimation relies strictly on the available information set \(\mathcal{F}_k\) (the past noisy trajectory). While the offline Hankel matrix contains future data correlations, the online estimator must remain causal; therefore, \(\hat{\eta}_1(k)\) is derived solely from \(\mathbf{z}_{\text{ini},m}(k)\) without utilizing future block rows, which are unavailable at time \(k\).}
\end{proof}

\textcolor{black}{
\begin{remark}[Estimation and Robustness of $\Sigma_\zeta$]
While $\Sigma_\zeta$ is assumed known, in practice it can be estimated via the sample covariance of steady-state sensor data or through an iterative feasible generalized least squares (FGLS) approach using residuals. The proposed estimator $\hat{\eta}_1(k)$ exhibits inherent robustness to estimation errors as explained below:
\begin{enumerate}
\item \textbf{Scale Invariance:} The estimator is invariant to absolute scaling of the covariance matrix (i.e., replacing $\Sigma_\zeta$ with $\alpha \Sigma_\zeta$ for $\alpha > 0$ yields the identical gain $\mathcal{K}$). Consequently, for MIMO systems the exact noise magnitude estimation is not required; only the relative variance ratios between sensor channels are necessary.
\item \textbf{Unbiasedness:} An incorrect covariance structure results in a loss of minimum-variance optimality (efficiency) but preserves the unbiasedness of the predictor, as the projection condition holds for any positive definite weighting matrix.
\end{enumerate}
\end{remark}
}

\subsection{Future Trajectory Parameterization}
\textcolor{black}{
To determine the future component of the regularization vector, denoted as $\mathbf{g}_f(k)$, any valid solution must lie in the null-space of $\mathbf{Z}_p$ to satisfy the orthogonality condition $\mathbf{g}_f \perp \mathbf{g}_{\text{ini}}$. Consequently, $\mathbf{g}_f(k)$ can be parameterized as:
\begin{equation}
    \mathbf{g}_f(k) = V_2 x_{np}(k),
    \label{eq:g_f_x_np}
\end{equation} 
where $V_2$ is the matrix of right singular vectors corresponding to the null space of $\mathbf{Z}_p$, derived from the SVD in \eqref{eq:Zp_SVD_L1V1T}. The vector $x_{np}(k) \in \mathbb{R}^{M - r}$ represents the degrees of freedom in the null space.}

The future input-output dynamics satisfy the fundamental data equation:
\begin{equation}
    \mathbf{Z}_f \Big(\mathbf{g}_{\text{ini}}(k) + \mathbf{g}_f(k)\Big) = \begin{bmatrix} \mathbf{u}_N(k) \\ \mathbf{y}_N(k) \end{bmatrix}.
    \label{eq:Zf_g_UYN}
\end{equation}
\textcolor{black}{Substituting \eqref{eq:g_f_x_np} into the future term $\mathbf{Z}_f \mathbf{g}_f(k)$ yields $\mathbf{Z}_f V_2 x_{np}(k)$. To compress this representation, we perform an SVD on the matrix product $\mathbf{Z}_f V_2$:}
\begin{equation}\label{eq:svd_ZfV2}
    \mathbf{Z}_f V_2 = W_f \Sigma_f V_f^T = \begin{bmatrix} W_{f_1} & W_{f_2} \end{bmatrix} \begin{bmatrix} \Sigma_{f_1} & 0 \\ 0 & 0 \end{bmatrix} \begin{bmatrix} V_{f_1}^T \\ V_{f_2}^T \end{bmatrix} = W_{f_1} \Sigma_{f_1} V_{f_1}^T.
\end{equation}
Since $\mathbf{Z}_f V_2$ represents the free future evolution consistent with past data, its dominant singular directions capture the controllable future dynamics.
Let $L_{f_1} = W_{f_1} \Sigma_{f_1}$. We define a reduced coordinate vector $\eta_2(k) \in \mathbb{R}^{n_u N}$ as $\eta_2(k) = V_{f_1}^T x_{np}(k)$, \textcolor{black}{ where under sufficient excitation of future inputs, the rank of $\mathbf{Z}_f V_2$ equals $n_u N$. This allows us to express the contribution of the future regularization vector compactly as:}
\begin{equation}
    \mathbf{Z}_f \mathbf{g}_f(k) = L_{f_1} \eta_2(k).
    \label{eq:Zf_gf_eta2}
\end{equation}
\textcolor{black}{In the presence of noise, the matrix $\Sigma_f$ will be full rank, taking the form $\Sigma_f = \text{diag}(\Sigma_{f_1}, \Sigma_{f_2})$. However, consistent with our noise-handling strategy, we assume the dominant singular values in $\Sigma_{f_1}$ capture the system dynamics, while $\Sigma_{f_2}$ is dominated by noise and is therefore truncated.}

\textcolor{black}{Note that $\eta_2(k) \in \mathbb{R}^{n_u N}$ and $L_{f_1} \in \mathbb{R}^{(n_u + n_y)N \times n_u N}$. The matrix $L_{f_1}$ is partitioned as $L_{f_1} = \begin{bmatrix} L_{f_u}^\top & L_{f_y}^\top \end{bmatrix}^\top$ with $L_{f_u} \in \mathbb{R}^{n_u N \times n_u N}$ and $L_{f_y} \in \mathbb{R}^{n_y N \times n_u N}$.}

\subsection{Dimension Reduction of the Future Trajectory Parameterization}
\textcolor{black}{We now unify the past and future parameterizations into a single low-order system. The past trajectory is given by $\mathbf{z}_{\text{ini}}(k) = \mathbf{Z}_p (\mathbf{g}_{\text{ini}}(k) + \mathbf{g}_f(k))$. Since $\mathbf{g}_f(k)$ lies in the null space of $\mathbf{Z}_p$, this simplifies to $\mathbf{z}_{\text{ini}}(k) = \mathbf{Z}_p \mathbf{g}_{\text{ini}}(k)$, which we previously derived in \eqref{eq:z_ini_eta} as $\mathbf{z}_{\text{ini}}(k) = L_1 \eta_1(k)$.}

\textcolor{black}{For the future trajectory, we substitute $\mathbf{g}_{\text{ini}}(k) = V_1 \eta_1(k)$ and the result from \eqref{eq:Zf_gf_eta2} into \eqref{eq:Zf_g_UYN}:
\begin{equation}
    \begin{bmatrix} \mathbf{u}_N(k) \\ \mathbf{y}_N(k) \end{bmatrix} = \mathbf{Z}_f V_1 \eta_1(k) + L_{f_1} \eta_2(k).
\end{equation}
Defining the coupling matrix $S = \mathbf{Z}_f V_1$, we can express the full trajectory vector in terms of the latent variables $\eta_1(k)$ and $\eta_2(k)$:}
\begin{equation}
    \begin{bmatrix}
        \mathbf{u}_{\text{ini}}(k) \\
        \mathbf{y}_{\text{ini}}(k) \\
        \mathbf{u}_N(k) \\
        \mathbf{y}_N(k)
    \end{bmatrix} =
    \begin{bmatrix}
        L_1 & 0 \\
        S & L_{f_1}
    \end{bmatrix}
    \begin{bmatrix}
        \eta_1(k) \\ \eta_2(k)
    \end{bmatrix}.
    \label{eq:uy_LS_eta}
\end{equation}
\textcolor{black}{To facilitate the predictor derivation, we partition the matrices $L_1$, $S$, and $L_{f_1}$ to separate input and output components:
\begin{equation}
    L_1 = \begin{bmatrix} L_u \\ L_y \end{bmatrix}, \quad
    S = \begin{bmatrix} S_{u} \\ S_{y} \end{bmatrix}, \quad
    L_{f_1} = \begin{bmatrix} L_{f_u} \\ L_{f_{y}} \end{bmatrix}.
\end{equation}
\noindent where the dimensions of the partitioned matrices are defined consistent with the signal dimensions: $L_u \in \mathbb{R}^{n_u T_{ini} \times r}$ and $L_y \in \mathbb{R}^{n_y T_{ini} \times r}$, $S_u \in \mathbb{R}^{n_u N \times r}$, $S_y \in \mathbb{R}^{n_y N \times r}$. This implicitly satisfies $\mathbf{u}_{\text{ini}}(k) = L_u \eta_1(k)$ and $\mathbf{y}_{\text{ini}}(k) = L_y \eta_1(k)$.}

\subsection{Optimal Multi-Step Predictor}
\textcolor{black}{We apply the parameterization in \eqref{eq:uy_LS_eta} to construct an $N$-step ahead predictor. Our goal is to derive the minimum variance unbiased predictor of the future output $\mathbf{y}_N(k)$ given the estimated initial state and future control inputs.
Expanding the lower block of \eqref{eq:uy_LS_eta}:
\begin{equation}
    \begin{split}
        \mathbf{u}_N(k) &= S_u \eta_1(k) + L_{f_u} \eta_2(k), \\
        \mathbf{y}_N(k) &= S_y \eta_1(k) + L_{f_y} \eta_2(k).
    \end{split}
    \label{eq:uyN_SL_eta}
\end{equation}
Substituting the black estimate $\hat{\eta}_1(k)$ from Theorem \ref{thm:black_eta1} into \eqref{eq:uyN_SL_eta} yields the prediction model:}
\begin{equation}\label{eq:uyN_predictor}
    \begin{split}
        \mathbf{u}_N(k) &= S_u \hat{\eta}_1(k) + L_{f_u} \eta_2(k), \\
        \mathbf{y}_N(k) &= S_y \hat{\eta}_1(k) + L_{f_y} \eta_2(k).
    \end{split}
\end{equation}

\textcolor{black}{Since $L_{f_u}$ is invertible, we can uniquely solve for the future latent vector $\eta_2(k)$ using the specified future inputs $\mathbf{u}_N(k)$:}
\begin{equation}
    \eta_2(k) = L_{f_u}^{-1} \left( \mathbf{u}_N(k) - S_u \hat{\eta}_1(k) \right).
    \label{eq:eta2_solution}
\end{equation}
Substituting \eqref{eq:eta2_solution} into the output equation in \eqref{eq:uyN_predictor}, we obtain the explicit predictor $\hat{\mathbf{y}}_N(k)$:
\begin{equation}\label{eq:yN_predictor}
    \begin{split}
        \hat{\mathbf{y}}_N(k) &= S_y \hat{\eta}_1(k) + L_{f_y} L_{f_u}^{-1} \left( \mathbf{u}_N(k) - S_u \hat{\eta}_1(k) \right) \\
        &= \left( S_y - L_{f_y} L_{f_u}^{-1} S_u \right) \hat{\eta}_1(k) + L_{f_y} L_{f_u}^{-1} \mathbf{u}_N(k).
    \end{split}
\end{equation}
\textcolor{black}{Replacing $\hat{\eta}_1(k)$ with the estimator $\mathcal{K}~ \mathbf{z}_{\text{ini},m}(k)$ derived in the Theorem \ref{thm:black_eta1}:}
\begin{equation}\label{eq:yN_predictor_final_form}
    \hat{\mathbf{y}}_N(k) = \left( S_y - L_{f_y} L_{f_u}^{-1} S_u \right) \mathcal{K} \mathbf{z}_{\text{ini},m}(k) + L_{f_y} L_{f_u}^{-1} \mathbf{u}_N(k).
\end{equation}
We define the predictor system matrices as:
\begin{equation}
    \mathbf{P}_1 = (S_y - L_{f_y} L_{f_u}^{-1} S_u) \mathcal{K}, \quad \mathbf{P}_2 = L_{f_y} L_{f_u}^{-1},
    \label{eq:P1_P2_def}
\end{equation}
\textcolor{black}{where $P_1 \in \mathbb{R}^{ n_y N \times (n_u + n_y)T_{ini}} $, $P_2 \in \mathbb{R}^{n_y N \times n_u N}$. Then, it leads to the compact affine prediction model:}
\begin{equation}
    \hat{\mathbf{y}}_N(k) = \mathbf{P}_1 \mathbf{z}_{\text{ini},m}(k) + \mathbf{P}_2 \mathbf{u}_N(k).
    \label{eq:yN_predictor_compact}
\end{equation}

To ensure the mathematical rigor of the NTDPC formulation, we introduce the following assumptions and lemma.

\begin{assumption}\label{ass:noise_properties}\textcolor{black}{[Future Excitation and Reduced-Order Parameterization]
The future input sequence $u_N$ is sufficiently exciting such that
\[
\operatorname{rank}(Z_f V_2) = n_u N,
\]
where $V_2$ spans the nullspace of $Z_p$.
As a result, the reduced future latent variable $\eta_2 \in \mathbb{R}^{n_u N}$ provides a complete parameterization of all admissible future trajectories consistent with the past data.}
\end{assumption}

\begin{remark}
    The reduced future input mapping $L_{f_u} \in \mathbb{R}^{n_u N \times n_u N}$ is assumed to be full rank in the noise-free case, which is guaranteed provided the offline training data satisfies the condition of persistency of excitation. In the presence of noise or numerical degeneracy, a regularized pseudoinverse of $L_{f_u}$ is employed to compute the future latent variable $\eta_2$. This guarantees the existence and uniqueness of the multi-step predictor in practice.
\end{remark}

\begin{assumption} \label{ass:constraints_feasibility}
    The constraint sets $\mathcal{U} \subseteq \mathbb{R}^{n_u N}$ and $\mathcal{Y} \subseteq \mathbb{R}^{n_y N}$ are convex, closed, and non-empty.
\end{assumption}

\begin{lemma} \label{lem:predictor_unbiased}
    \textcolor{black}{Under Assumptions  \ref{ass:noise_properties} and invertible property of $L_{f_u}$, the predictor in (\ref{eq:yN_predictor_compact}) is an unbiased estimator of the true future output $\mathbf{y}_N^{\text{true}}(k)$ given the parameters $\eta_1(k)$ and $\mathbf{u}_N(k)$. That is, $\mathbb{E}[\hat{\mathbf{y}}_N(k) \mid \eta_1(k), \mathbf{u}_N(k)] = \mathbf{y}_N^{\text{true}}(k)$.}
\end{lemma}
\textcolor{black}{
\begin{proof}
    The true output of the system, given the latent state $\eta_1(k)$ and input $\mathbf{u}_N(k)$, is defined by the second line of \eqref{eq:uyN_SL_eta}: $\mathbf{y}_N^{\text{true}}(k) = S_y \eta_1(k) + L_{f_y} \eta_2(k)$. Substituting the exact solution for $\eta_2(k)$, we get:
    \begin{equation}
        \mathbf{y}_N^{\text{true}}(k) = \left( S_y - L_{f_y} L_{f_u}^{-1} S_u \right) \eta_1(k) + L_{f_y} L_{f_u}^{-1} \mathbf{u}_N(k).
    \end{equation}
    Taking the expectation of the predictor \eqref{eq:yN_predictor_final_form}, and noting that $\mathbf{u}_N(k)$ is deterministic and $\mathbb{E}[\hat{\eta}_1(k) \mid \eta_1(k)] = \eta_1(k)$ (from Theorem \ref{thm:black_eta1}):
    \begin{equation}
        \begin{split}
            \mathbb{E}[\hat{\mathbf{y}}_N(k)] &= \left( S_y - L_{f_y} L_{f_u}^{-1} S_u \right) \mathbb{E}[\hat{\eta}_1(k)] + L_{f_y} L_{f_u}^{-1} \mathbf{u}_N(k) \\
            &= \left( S_y - L_{f_y} L_{f_u}^{-1} S_u \right) \eta_1(k) + L_{f_y} L_{f_u}^{-1} \mathbf{u}_N(k) \\
            &= \mathbf{y}_N^{\text{true}}(k).
        \end{split}
    \end{equation}
    This confirms the predictor is unbiased. Since $\mathbf{P}_1$ and $\mathbf{P}_2$ rely only on the SVD of the offline data, they can be computed offline.
\end{proof}
}

We now formally state the NTDPC optimization problem.

\vspace{3mm}
\textbf{Problem 4 [NTDPC]}:
Given past input and measured output trajectories, NTDPC computes the optimal control sequence by solving the following convex optimization problem at each time step $k$:
\begin{align}
    \min_{\mathbf{u}_N(k), \mathbf{y}_N(k), \sigma_y(k)} \quad & \|\mathbf{y}_N(k)\|_Q^2 + \|\mathbf{u}_N(k)\|_R^2 + \|\sigma_y(k)\|_{\Lambda_y}^2 \notag \\
    \text{subject to} \quad & \mathbf{y}_N(k) = \mathbf{P}_1 \mathbf{z}_{\text{ini},m}(k) + \mathbf{P}_2 \mathbf{u}_N(k) + \sigma_y(k), \label{eq:NTDPC_constraint} \\
    & \mathbf{u}_N(k) \in \mathcal{U}, \quad \mathbf{y}_N(k) \in \mathcal{Y}, \notag
\end{align}
\textcolor{black}{where $\mathbf{y}_N(k) \in \mathbb{R}^{n_y N}$ represents the predicted future output used in the cost function, and $\sigma_y(k) \in \mathbb{R}^{n_y N}$ is a slack variable ensuring feasibility against noise and model mismatch. The weighting matrices $Q, R, \Lambda_y \succ 0$ are user-defined. The vector \textcolor{black}{$\mathbf{z}_{\text{ini},m}(k)$} contains the most recent $T_{\text{ini}}$ input and noisy output measurements. Algorithm 1 summarizes the implementation steps.}
\textcolor{black}{
\begin{remark}
    To prevent numerical bias during the SVD factorization of $\mathbf{Z}_p$ and $\mathbf{Z}_f V_2$, it is recommended to normalize the input and output data. This balances the singular values when inputs and outputs have significantly different magnitudes. We apply the scaling:
    \begin{equation}
        \mathbf{u}(k) \leftarrow M_u^{-1} \mathbf{u}(k), \quad \mathbf{y}(k) \leftarrow M_y^{-1} \mathbf{y}(k),
    \end{equation}
    where $M_u = \text{diag}(m_{u_1}, \dots, m_{u_{n_u}})$ and $M_y = \text{diag}(m_{y_1}, \dots, m_{y_{n_y}})$. In this study, the scaling factors are chosen as the mean absolute values of the signals: $m_{u_i} = \text{mean}(|\mathbf{u}_i|)$ and $m_{y_i} = \text{mean}(|\mathbf{y}_i|)$. The computed optimal control input $\mathbf{u}_N^*(k)$ must be denormalized before being applied to the plant.
\end{remark}
}

\begin{algorithm}
\caption{Implementation of NTDPC}
\label{alg:dp_constrained_consensus}
\begin{algorithmic}[1]
\REQUIRE Generate P.E input and measure output 
\STATE Normalize the input and output signals 
\STATE Construct $\mathbf{U}_p, \mathbf{Y}_p, \mathbf{U}_f, \mathbf{Y}_f$ 
\STATE Compute and extract $L_1,V_1,V_2,L_{f_1}$ using SVD
\STATE Calculate the matrices $\mathbf{P}_1,\mathbf{P}_2$
\STATE \textcolor{black}{Build $\mathbf{z}_{\text{ini},m}(k)$  using the past $T_{\text{ini}}$ data samples}
\STATE Obtain $\hat{\mathbf{y}}_N (k)= \mathbf{P}_1~{\mathbf{z}}_{\text{ini},m}(k) + \mathbf{P}_2~\mathbf{u}_N(k)$
\STATE Solve the QP problem (Problem 4) and obtain $\mathbf{u}_N^*(k)$
\STATE \textcolor{black}{Extract optimal first input $u^*(k)$ from $u_N^*(k)$}
\STATE \textcolor{black}{Denormalize $u^*(k) \leftarrow M_u u^*(k)$ and apply to system}
\end{algorithmic}
\end{algorithm} 
\vspace{-3mm}

\section{Stability Analysis}
\label{sec:stability}

This section establishes a rigorous stability analysis of the developed NTDPC approach. \textcolor{black}{We examine the system's behavior in both noise-free and noisy scenarios. In the noise-free case, we prove asymptotic stability, ensuring exact tracking of the reference. In the noisy case, we demonstrate Input-to-State Stability (ISS), guaranteeing robustness to bounded measurement noise. To facilitate this analysis, we define the following key concepts and notations, adapted from the dissipativity framework for data-driven control \cite{lazar2021dissipativity}.}
\textcolor{black}{
\begin{definition}[Stage Cost]\label{def:stage_cost}
    Assume the weighting matrices \( Q \) and \( R \) are block-diagonal with blocks \( Q_i \succ 0 \) and \( R_i \succ 0 \). The stage cost \( l(\cdot, \cdot) \) is defined as:
    \begin{equation}
        l(\mathbf{y}, \mathbf{u}) = \|\mathbf{y} - r_y\|_{Q_i}^2 + \|\mathbf{u} - r_u\|_{R_i}^2,
    \end{equation}
    where \( r_y \in \mathbb{R}^{n_y} \) and \( r_u \in \mathbb{R}^{n_u} \) denote the steady-state output and input references, respectively.
    The total cost function over the prediction horizon \( N \) is given by:
    \begin{equation}
        J^*(k) = \sum_{i=0}^{N-1} l(\mathbf{y}^*(i|k), \mathbf{u}^*(i|k)) + \|\sigma_y^*(k)\|_{\Lambda_y}^2,
    \end{equation}
    where \( \mathbf{y}^*(i|k) \) and \( \mathbf{u}^*(i|k) \) are the predicted output and input at step \( i \) given information at time \( k \).
\end{definition} 
}
\textcolor{black}{
\begin{definition}[Non-Minimal State] \label{def:non_minimal_state}
    The non-minimal state vector $\mathbf{z}_{\text{ini}}(k)$ is defined as:
\[
\mathbf{z}_{\text{ini}}(k) = \operatorname{col}(\mathbf{u}_{\text{ini}}(k), \mathbf{y}_{\text{ini}}(k)).
\]
The error state $\tilde{\mathbf{z}}_{\text{ini}}(k)$ represents the deviation from equilibrium:
\[
\tilde{\mathbf{z}}_{\text{ini}}(k) = \operatorname{col}(\mathbf{u}_{\text{ini}}(k) - \mathbf{r}_u, \mathbf{y}_{\text{ini}}(k) - \mathbf{r}_y),
\]
where $\mathbf{r}_u \in \mathbb{R}^{n_u T_{\text{ini}}}$ and $\mathbf{r}_y \in \mathbb{R}^{n_y T_{\text{ini}}}$ are stacked reference vectors.
\end{definition}
}

\begin{definition}[Class \(\mathcal{K}\) and \(\mathcal{KL}\) Functions]\label{def:class_k_kl}
    A continuous function \(\varphi : \mathbb{R}_+ \to \mathbb{R}_+\) belongs to class \(\mathcal{K}\) if it is strictly increasing and \(\varphi(0) = 0\). It belongs to class \(\mathcal{K}_{\infty}\) if it is also unbounded, i.e., \(\lim_{s \to \infty} \varphi(s) = \infty\). A function \(\beta : \mathbb{R}_+ \times \mathbb{R}_+ \to \mathbb{R}_+\) belongs to class \(\mathcal{KL}\) if \(\beta(\cdot, t) \in \mathcal{K}\) for each fixed \( t \geq 0 \), and \(\lim_{t \to \infty} \beta(s, t) = 0\) for each fixed \( s > 0 \). The identity function is denoted by \( \operatorname{id}(s) = s \).
\end{definition}

\begin{definition}[Filtration]\label{def:filtration}
    The filtration \( \mathcal{F}_k \) is the \(\sigma\)-algebra generated by the history of measurements \( \{\mathbf{z}_{\text{ini},m}(t)\}_{t=0}^k \), representing the information available to the controller up to time \( k \).
\end{definition}

\subsection{Noise-Free Case}
\textcolor{black}{We first analyze the stability of the NTDPC scheme in the ideal noise-free scenario (\( \mathbf{\zeta}_{\text{ini}}(k) = 0 \)) while tracking constant admissible references \( (r_y, r_u) \).}

\begin{assumption}[Positive Definiteness]\label{ass:positive_definiteness_cost}
    There exist functions \(\alpha_1, \alpha_2 \in \mathcal{K}_{\infty}\) such that for all admissible pairs \((\mathbf{y}, \mathbf{u})\):
    \begin{equation}
        \alpha_1(\|\operatorname{col}(\mathbf{y} - r_y, \mathbf{u} - r_u)\|) \leq l(\mathbf{y}, \mathbf{u}) \leq \alpha_2(\|\operatorname{col}(\mathbf{y} - r_y, \mathbf{u} - r_u)\|).
    \end{equation}
\end{assumption}

\begin{assumption}[Terminal Stabilizing Condition]\label{ass:terminal_condition_nf}
    For any admissible initial state \( \mathbf{z}_{\text{ini}}(k) \), there exists a prediction horizon \( N \geq T_{\text{ini}} \) and a terminal control law satisfying a contraction property. Specifically, there exists a function \(\rho \in \mathcal{K}_{\infty}\) with \(\rho < \operatorname{id}\) such that the terminal cost satisfies:
    \begin{equation}
        l(\hat{\mathbf{y}}(N|k), \hat{\mathbf{u}}(N|k)) \leq (\operatorname{id} - \rho) \circ l(\mathbf{y}(k - T_{\text{ini}}), \mathbf{u}(k - T_{\text{ini}})).
    \end{equation}
    \textcolor{black}{This condition ensures that the cost of the appended terminal state is strictly less than the cost of the oldest data point being discarded from the history window.}
\end{assumption}

\begin{assumption}[Feasibility]\label{ass:feasibility_nf}
    \textcolor{black}{The NTDPC optimization (Problem 4) is recursively feasible for all \( k \geq T_{\text{ini}} \). In the noise-free case, the optimal slack variable satisfies \( \sigma_y^*(k) = 0 \). Furthermore, the equilibrium pair \( (r_y, r_u) \) lies within the interior of the constraint sets \( \mathcal{Y} \times \mathcal{U} \).}
\end{assumption}

\begin{theorem}\label{thm:stability_nf}
    Under Assumptions \ref{ass:positive_definiteness_cost}, \ref{ass:terminal_condition_nf}, and \ref{ass:feasibility_nf}, \textcolor{black}{consider the closed-loop system controlled by NTDPC in the absence of noise (\( \mathbf{\zeta}_{\text{ini}}(k) = 0 \)). The equilibrium \( (r_y, r_u) \) is asymptotically stable, implying \( \lim_{k \to \infty} \mathbf{y}(k) = r_y \) and \( \lim_{k \to \infty} \mathbf{u}(k) = r_u \).}
\end{theorem}

\begin{proof}
    \textcolor{black}{We construct a Lyapunov candidate using the optimal cost function augmented with a history storage function. Let the optimal cost at time \( k \) be \( J^*(k) \). Since \( \sigma_y^*(k) = 0 \) (Assumption \ref{ass:feasibility_nf}), \( J^*(k) = \sum_{i=0}^{N-1} l(\mathbf{y}^*(i|k), \mathbf{u}^*(i|k)) \).}
    
   \textcolor{black}{ Define the auxiliary history function \( W(\mathbf{z}_{\text{ini}}(k)) \) as the accumulated cost of the past \( T_{\text{ini}} \) steps:
    \begin{equation}
        W(\mathbf{z}_{\text{ini}}(k)) = \sum_{i=1}^{T_{\text{ini}}} l(\mathbf{y}(k - i), \mathbf{u}(k - i)).
    \end{equation}
    Consider the Lyapunov function candidate:
    \begin{equation}
        V(\mathbf{z}_{\text{ini}}(k)) = J^*(k) + W(\mathbf{z}_{\text{ini}}(k)).
    \end{equation}}
   \textcolor{black}{ 
    Using Assumption \ref{ass:positive_definiteness_cost}, the history function is bounded below by:
    \begin{equation}
        W(\mathbf{z}_{\text{ini}}(k)) \geq \sum_{i=1}^{T_{\text{ini}}} \alpha_1(\|\operatorname{col}(\mathbf{y}(k - i) - r_y, \mathbf{u}(k - i) - r_u)\|).
    \end{equation}}
    \textcolor{black}{Since the error state \( \tilde{\mathbf{z}}_{\text{ini}}(k) \) is composed of these deviations, there exists \( \alpha_{1,V} \in \mathcal{K}_{\infty} \) such that \( V(\mathbf{z}_{\text{ini}}(k)) \geq W(\mathbf{z}_{\text{ini}}(k)) \geq \alpha_{1,V}(\|\tilde{\mathbf{z}}_{\text{ini}}(k)\|) \).
    Similarly, using the upper bound \( \alpha_2 \) and the linearity of the predictor \( \hat{\mathbf{y}}_N(k) = \mathbf{P}_1 \mathbf{z}_{\text{ini}}(k) + \mathbf{P}_2 \mathbf{u}_N(k) \), the optimal cost \( J^*(k) \) is bounded by the state magnitude. Thus, there exists \( \alpha_{2,V} \in \mathcal{K}_{\infty} \) such that \( V(\mathbf{z}_{\text{ini}}(k)) \leq \alpha_{2,V}(\|\tilde{\mathbf{z}}_{\text{ini}}(k)\|) \).}    
    
    At time \( k+1 \), we construct a feasible suboptimal input sequence \( \mathbf{u}_s(k+1) \) by shifting the optimal sequence from time \( k \) and appending the terminal law \( \hat{\mathbf{u}}(N|k) \):
    \begin{equation}
        \mathbf{u}_s(k+1) = \operatorname{col}(\mathbf{u}^*(1|k), \ldots, \mathbf{u}^*(N-1|k), \hat{\mathbf{u}}(N|k)).
    \end{equation}
    The corresponding cost is \( J_s(k+1) = J^*(k) - l(\mathbf{y}^*(0|k), \mathbf{u}^*(0|k)) + l(\hat{\mathbf{y}}(N|k), \hat{\mathbf{u}}(N|k)) \).
    By optimality, \( J^*(k+1) \leq J_s(k+1) \).
    
    \textcolor{black}{Now consider the evolution of the history function \( W \):
    \begin{align}
        W(\mathbf{z}_{\text{ini}}(k+1)) &= \sum_{i=1}^{T_{\text{ini}}} l(\mathbf{y}(k + 1 - i), \mathbf{u}(k + 1 - i)) \notag \\
        &= W(\mathbf{z}_{\text{ini}}(k)) + l(\mathbf{y}(k), \mathbf{u}(k)) - l(\mathbf{y}(k - T_{\text{ini}}), \mathbf{u}(k - T_{\text{ini}})).
    \end{align}
    Note that in the noise-free prediction, \( \mathbf{y}^*(0|k) = \mathbf{y}(k) \) and \( \mathbf{u}^*(0|k) = \mathbf{u}(k) \).
    Computing the difference \( \Delta V(k) = V(\mathbf{z}_{\text{ini}}(k+1)) - V(\mathbf{z}_{\text{ini}}(k)) \):
    \begin{align}
        \Delta V(k) &\leq [J^*(k) - l(\mathbf{y}(k), \mathbf{u}(k)) + l(\hat{\mathbf{y}}(N|k), \hat{\mathbf{u}}(N|k))] \notag \\
        &\quad + [W(\mathbf{z}_{\text{ini}}(k)) + l(\mathbf{y}(k), \mathbf{u}(k)) - l(\mathbf{y}(k - T_{\text{ini}}), \mathbf{u}(k - T_{\text{ini}}))] \notag \\
        &\quad - [J^*(k) + W(\mathbf{z}_{\text{ini}}(k))] \notag \\
        &= l(\hat{\mathbf{y}}(N|k), \hat{\mathbf{u}}(N|k)) - l(\mathbf{y}(k - T_{\text{ini}}), \mathbf{u}(k - T_{\text{ini}})).
    \end{align}
    Applying Assumption \ref{ass:terminal_condition_nf} (Terminal Stabilizing Condition):
    \begin{equation}
        \Delta V(k) \leq -\rho \circ l(\mathbf{y}(k - T_{\text{ini}}), \mathbf{u}(k - T_{\text{ini}})).
    \end{equation}
    Since \( l(\cdot) \) is positive definite with respect to the deviation from equilibrium, this implies \( \Delta V(k) \) is negative definite in terms of the error state. Standard Lyapunov arguments \cite{jiang2001input} then guarantee that \( \|\tilde{\mathbf{z}}_{\text{ini}}(k)\| \to 0 \) as \( k \to \infty \), confirming asymptotic stability.}
\end{proof}

We assume that the system is stabilizable and that an admissible input sequence exists that maintains feasibility beyond the prediction horizon.

\subsection{Noisy Case: Robustness to Measurement Noise}

We now extend the analysis to the case where the measurements are corrupted by noise (\(\mathbf{\zeta}_{\text{ini}}(k) \neq 0\)). We prove that the NTDPC closed-loop system is Input-to-State Stable (ISS) with respect to the measurement noise \(\mathbf{\zeta}_{\text{ini}}(k)\), utilizing the stochastic stability framework for MPC \cite{jiang2001input, verheijen2023handbook}.

\begin{assumption}[Bounded Noise Variance]\label{ass:bounded_noise}
    \textcolor{black}{The measurement noise is a martingale difference sequence with bounded variance. Specifically, there exists a constant \(\bar{c} > 0\) such that for all \( k \geq 0 \):
    \begin{equation}
        \mathbb{E}[\|\mathbf{\zeta}_{\text{ini}}(k)\|^2 \mid \mathcal{F}_k] \leq \bar{c},
    \end{equation}
    where \(\mathcal{F}_k\) denotes the filtration generated by measurements up to time \(k\).}
\end{assumption}

\begin{assumption}\textcolor{black}{[Lipschitz Continuity of the Predictor]\label{ass:lipschitz_predictor}
    The predictor mapping is Lipschitz continuous with respect to the initial condition. There exists a constant \(L_P > 0\) such that for any two initial conditions \(\mathbf{z}_1, \mathbf{z}_2\):
    \begin{equation}
        \|\mathbf{P}_1 (\mathbf{z}_1 - \mathbf{z}_2)\| \leq L_P \|\mathbf{z}_1 - \mathbf{z}_2\|.
    \end{equation}
    Consequently, the deviation in the slack variable due to noise is linearly bounded:
    \begin{equation}
        \|\sigma_{y,s}(k)\| \leq L_P \|\mathbf{\zeta}_{\text{ini}}(k)\|,
    \end{equation}
    where \(\mathbf{z}_{\text{ini},m}(k) = \mathbf{z}_{\text{ini}}(k) + \operatorname{col}(0, \mathbf{\zeta}_{\text{ini}}(k))\).}
\end{assumption}

\begin{assumption}[Stochastic Terminal Stabilizing Condition]\label{ass:terminal_condition_noisy}
    For any admissible state, there exists a terminal control law \(\hat{\mathbf{u}}(N|k)\) and a class \(\mathcal{K}_\infty\) function \(\rho < \operatorname{id}\) such that the expected terminal cost satisfies a contraction:
    \begin{equation}
        \mathbb{E}[l(\hat{\mathbf{y}}(N|k), \hat{\mathbf{u}}(N|k)) \mid \mathcal{F}_k] \leq (\operatorname{id} - \rho) \circ \mathbb{E}[l(\mathbf{y}(k - T_{\text{ini}}), \mathbf{u}(k - T_{\text{ini}})) \mid \mathcal{F}_k].
    \end{equation}
\end{assumption}

\begin{assumption}\textcolor{black}{[Recursive Feasibility and Slack Bound]\label{ass:feasibility_noisy}
Problem 4 is feasible with probability at least $1-\epsilon$. Furthermore, provided the slack penalty weighting matrix $\Lambda_y$ is chosen sufficiently large relative to the noise covariance, the optimal slack variable $\sigma_y^*(k)$ satisfies a quadratic bound with respect to the noise realization:
\begin{equation}
    \mathbb{E}[||\sigma_y^*(k)||_{\Lambda_y}^2 | \mathcal{F}_k] \le c_\sigma \mathbb{E}[||\zeta_{ini}(k)||^2 | \mathcal{F}_k],
\end{equation}
for some constant $c_\sigma > 0$.}
\end{assumption}

\begin{remark}
    Assumption \ref{ass:terminal_condition_noisy} is typically satisfied by selecting \(\hat{\mathbf{u}}(N|k)\) as the steady-state reference \(r_u\), utilizing the inherent robustness of the predictor established in Lemma \ref{lem:predictor_unbiased}. Assumption \ref{ass:feasibility_noisy} is empirically supported by the sensitivity index analysis ($I_s \leq 0.7$), which ensures sufficient separability between signal and noise subspaces to maintain feasibility.
\end{remark}

\begin{theorem}\label{thm:stability_noisy}
    Under Assumptions \ref{ass:positive_definiteness_cost}, \ref{ass:bounded_noise}, \ref{ass:lipschitz_predictor}, \ref{ass:terminal_condition_noisy}, and \ref{ass:feasibility_noisy}, the closed-loop system generated by NTDPC is Input-to-State Stable (ISS) with respect to the noise sequence \(\mathbf{\zeta}_{\text{ini}}(k)\). Specifically, there exist functions \(\beta \in \mathcal{KL}\) and \(\gamma \in \mathcal{K}_{\infty}\) such that:
    \begin{equation}
        \mathbb{E}[\|\tilde{\mathbf{z}}_{\text{ini}}(k)\|] \leq \beta(\|\tilde{\mathbf{z}}_{\text{ini}}(0)\|, k) + \gamma\left(\sup_{0 \leq t < k} \|\mathbf{\zeta}_{\text{ini}}(t)\|_{\infty}\right).
    \end{equation}
\end{theorem}

\begin{proof}
    \textcolor{black}{Consider the stochastic Lyapunov function candidate:
    \begin{equation}
        V(k) = \mathbb{E}[J^*(k) + W(\mathbf{z}_{\text{ini}}(k)) \mid \mathcal{F}_k].
    \end{equation}
    }
   \textcolor{black}{ 
    From the noise-free analysis, \(W(\mathbf{z}_{\text{ini}}(k))\) is lower-bounded by \(\alpha_{1,W}(\|\tilde{\mathbf{z}}_{\text{ini}}(k)\|)\).
    The optimal cost \(J^*(k)\) differs from the noise-free case by the slack term. Using Assumption \ref{ass:feasibility_noisy} and the boundedness of the noise:
    \begin{equation}
        \mathbb{E}[J^*(k) \mid \mathcal{F}_k] \leq \alpha_{2,J}(\|\tilde{\mathbf{z}}_{\text{ini}}(k)\|) + c_\sigma \bar{c}.
    \end{equation}
    Thus, \(V(k)\) satisfies the sandwich bounds:
    \begin{equation}
        \alpha_{1,V}(\|\tilde{\mathbf{z}}_{\text{ini}}(k)\|) \leq V(k) \leq \alpha_{2,V}(\|\tilde{\mathbf{z}}_{\text{ini}}(k)\|) + c_\sigma \bar{c}.
    \end{equation}
    }
    
   \textcolor{black}{
    We construct the standard suboptimal sequence \(\mathbf{u}_s(k+1)\) by shifting the optimal input and appending the terminal law. The slack variable for this suboptimal solution, \(\sigma_{y,s}(k+1)\), absorbs the prediction error due to noise. From Assumption \ref{ass:lipschitz_predictor}, its expected cost is bounded:
    \begin{equation}
        \mathbb{E}[\|\sigma_{y,s}(k+1)\|_{\Lambda_y}^2 \mid \mathcal{F}_k] \leq \lambda_{\max}(\Lambda_y) L_P^2 \bar{c}.
    \end{equation}
    Define \( \bar{C}_{\text{noise}} = (c_\sigma + \lambda_{\max}(\Lambda_y) L_P^2)\bar{c} \). The expected difference in the Lyapunov function is:
    \begin{align}
        \mathbb{E}[V(k+1) - V(k) \mid \mathcal{F}_k] &\leq \mathbb{E}\left[ l(\hat{\mathbf{y}}(N|k), \hat{\mathbf{u}}(N|k)) - l(\mathbf{y}(k-T_{\text{ini}}), \mathbf{u}(k-T_{\text{ini}})) \right. \notag \\
        &\quad \left. + l(\mathbf{y}(k), \mathbf{u}(k)) - l(\mathbf{y}^*(0|k), \mathbf{u}^*(0|k)) \mid \mathcal{F}_k \right] + \bar{C}_{\text{noise}}.
    \end{align}
    The first line is bounded by \(-\rho(\cdot)\) via Assumption \ref{ass:terminal_condition_noisy}.
    The crucial term is the difference in the second line: \( \Delta_{\text{stage}} = l(\mathbf{y}(k), \mathbf{u}(k)) - l(\mathbf{y}^*(0|k), \mathbf{u}^*(0|k)) \).
    In the noise-free case, \( \mathbf{y}(k) = \mathbf{y}^*(0|k) \), so this term vanishes. In the noisy case, the measured output \( \mathbf{y}(k) \) differs from the predicted output \( \mathbf{y}^*(0|k) \) due to the noise realization at time \( k \). Since \( l(\cdot) \) is locally Lipschitz with constant \( L_l \) on the compact constraint set:
    \begin{equation}
        \mathbb{E}[\Delta_{\text{stage}} \mid \mathcal{F}_k] \leq L_l \mathbb{E}[\|\mathbf{y}(k) - \mathbf{y}^*(0|k)\| \mid \mathcal{F}_k] \leq L_l \mathbb{E}[\|\mathbf{\zeta}(k)\| \mid \mathcal{F}_k] \leq L_l \sqrt{\bar{c}}.
    \end{equation}
    Combining these results:
    \begin{equation}
        \mathbb{E}[V(k+1) - V(k) \mid \mathcal{F}_k] \leq -\rho(\alpha_{4,V}(\mathbb{E}[\|\tilde{\mathbf{z}}_{\text{ini}}(k)\|])) + \bar{C}_{\text{total}},
    \end{equation}
    where \( \bar{C}_{\text{total}} \) aggregates all noise-dependent constants. This inequality confirms that the Lyapunov function decreases outside a bounded region determined by the noise level, which is the definition of ISS \cite{jiang2001input}.}
\end{proof}

\section{Numerical Simulations}

Consider the discrete-time linear model system we aim to control, characterized by the following matrices \cite{camacho2007introduction}, Section 6.4:
\begin{equation}
\begin{split}
& A =
\begin{bmatrix}
0.9997 & 0.0038 & -0.0001 & -0.0322 \\
-0.0056 & 0.9648 & 0.7446 & 0.0001 \\
0.0020 & -0.0097 & 0.9543 & -0.0000 \\
0.0001 & -0.0005 & 0.0978 & 1.0000
\end{bmatrix},
\quad\\&
B =
\begin{bmatrix}
0.0010 & 0.1000 \\
-0.0615 & 0.0183 \\
-0.1133 & 0.0586 \\
-0.0057 & 0.0029
\end{bmatrix},
C =
\begin{bmatrix}
1.0000 & 0 & 0 & 0 \\
0 & -1.0000 & 0 & 7.7400
\end{bmatrix}
\end{split}
\end{equation}
These matrices are obtained by discretizing the dynamics of the longitudinal motion of a Boeing 747 using a zero-order hold with a sampling time $T_s = 0.1 \text{ s}$. The system has two inputs: the throttle $u_1$ and the elevator angle $u_2$, and two outputs: the aircraft's longitudinal velocity and climb rate.

The control objective is to increase the longitudinal velocity ($y_1(t)$) from 0 ft/s to 10 ft/s while maintaining the climb rate ($y_2(t)$) at zero. This corresponds to a reference output $r_y = \begin{bmatrix} 10 & 0 \end{bmatrix}^T$. The inputs and outputs are constrained as follows:
\begin{equation}
\begin{split}
&U := \left\{ u \in \mathbb{R}^2 : \begin{bmatrix} -20 \\ -20 \end{bmatrix} \leq u \leq \begin{bmatrix} 20 \\ 20 \end{bmatrix} \right\},
\quad \\&
Y := \left\{ y \in \mathbb{R}^2 : \begin{bmatrix} -25 \\ -15 \end{bmatrix} \leq y \leq \begin{bmatrix} 25 \\ 15 \end{bmatrix} \right\}.
\end{split}
\end{equation}
 
In this simulation, we set $T_{\text{ini}} =N= 20$, data length for Hankel matrices $M = 2500$.

\subsection{Noise-free case}

The Fig.\ref{fig:Y11} illustrates the performance of NTDPC in the noise-free case ($\mathbf{\zeta}_{\text{ini}}(k) = 0$) when tracking constant admissible references ($(r_y, r_u)$), as analyzed in Section 5.1. The results are consistent with the asymptotic stability proven in Theorem \ref{thm:stability_nf}, where $y(k) \to r_y$ and $u(k) \to r_u$ as $k \to \infty$ under Assumptions \ref{ass:positive_definiteness_cost}, \ref{ass:terminal_condition_nf}, and \ref{ass:feasibility_nf}. The top plot displays the primary output $y_1(k)$ for NTDPC (black), SPC (red), and SMMPC (green) over time, with the dashed line representing the reference $r_y$; all methods converge to $r_y$, but NTDPC exhibits the fastest transient response and minimal overshoot, reflecting the effectiveness of the data-driven predictor and terminal condition. The second plot shows a secondary output $y_2(k)$, with similar convergence behavior, where NTDPC (black) stabilizes more quickly than SPC (red) and SMMPC (green), supporting the robustness of the NTDPC framework to initial conditions encoded in $\mathbf{z}_{\text{ini}}(k)$. The third plot presents the primary input $u_1(k)$, where NTDPC (black) adjusts rapidly from an initial value, stabilizing near $r_u$ (dashed line) with less oscillation compared to SPC (red) and SMMPC (green), consistent with the bounded input constraints in $\mathcal{U}$. Finally, the bottom plot depicts a secondary input $u_2(k)$, showing NTDPC (black) achieving a steady state closer to $r_u$ with reduced variability, further validating the stability properties derived from the Lyapunov function $V(\mathbf{z}_{\text{ini}}(k))$. Overall, the convergence of outputs and inputs to their respective references demonstrates the asymptotic stability of NTDPC, as predicted by the storage function analysis $V(\mathbf{z}_{\text{ini}}(k)) = J^*(k) + W(\mathbf{z}_{\text{ini}}(k))$, and the superior performance of NTDPC over SPC and SMMPC highlights the advantage of the noise-tolerant predictor and optimization framework, particularly in the absence of noise.
    \begin{figure}[!ht] 
\centering
\includegraphics[width=350pt]{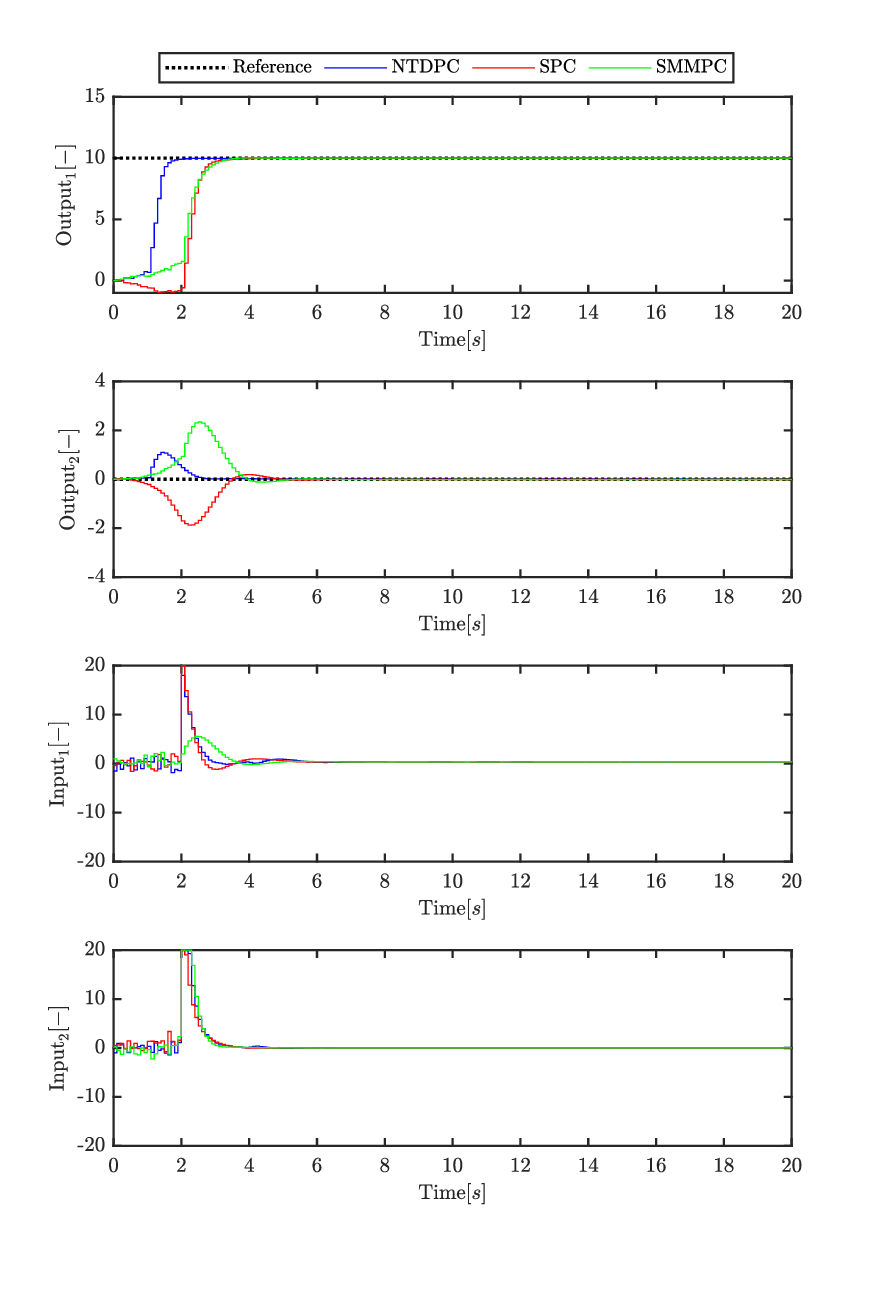}
\vspace{-1cm}
\caption{Output trajectories for noise-free case}
\label{fig:Y11}
\end{figure}

\subsection{Noisy case}
 In Fig. \ref{fig:Y1}, the simulation results demonstrate the performance of NTDPC in comparison to SPC and SMMPC under the noise effect, $\sigma_{\zeta}^2=0.25I$. While SMMPC exhibits aggressive and unstable control actions leading to poor regulation of outputs, NTDPC and SPC show reasonable results, despite an initial overshoot and undershoot, ultimately achieves stable tracking of the reference for both outputs with damped oscillations and less aggressive control inputs that settle over time. This balanced response in both outputs and the settling behavior of the control inputs highlight the effectiveness and stability of the NTDPC and SPC strategies compared to the oscillatory and poorly regulated responses observed with SMMPC. 
    \begin{figure}[!ht] 
\centering
\includegraphics[width=350pt]{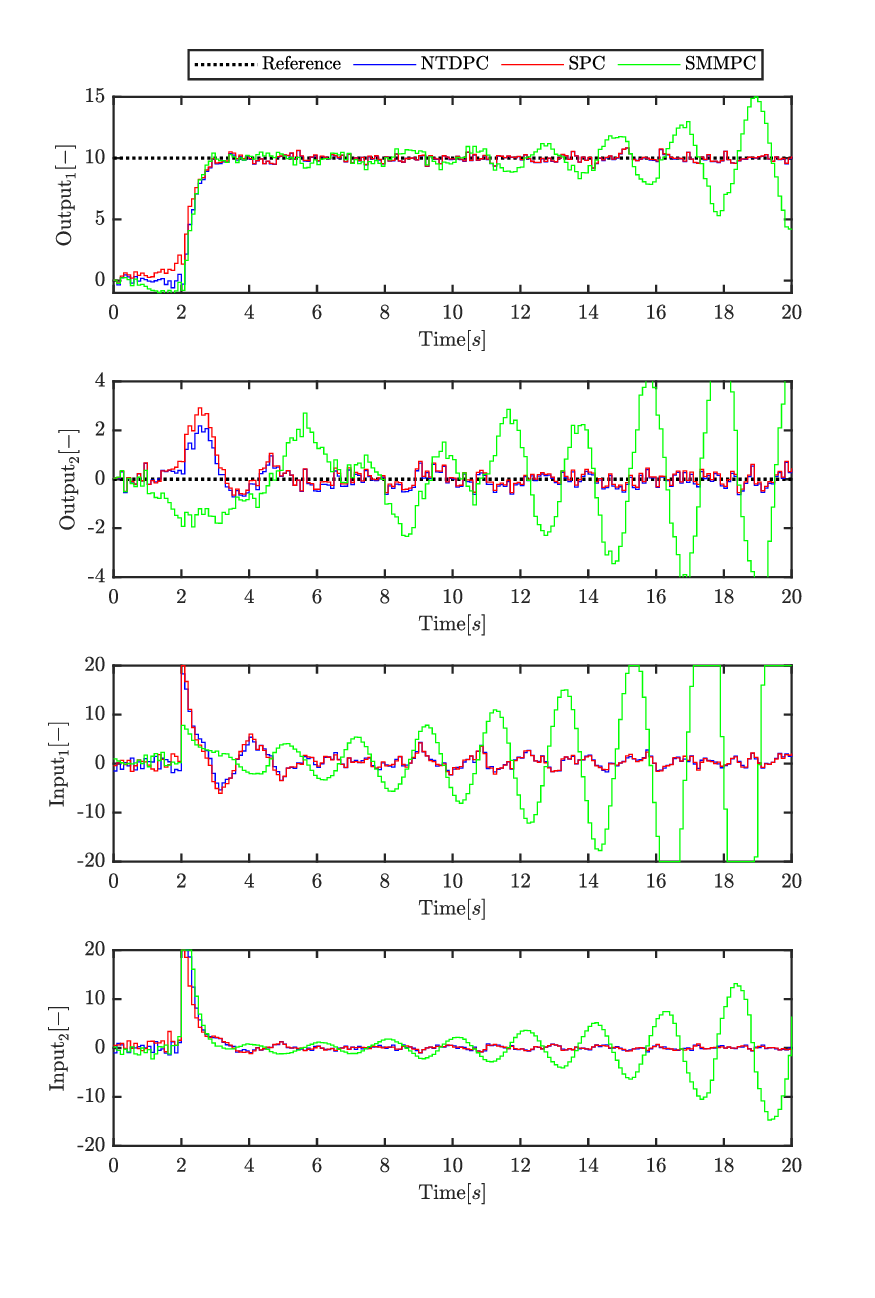}
\vspace{-1cm}
\caption{Output trajectories over 50 Monte Carlo runs for noisy case}
\label{fig:Y1}
\end{figure}

To quantify controller performance, the objective function is evaluated and plotted in Fig. \ref{fig:ISE2} which reveals the superiority of NTDPC and SPC compared to  SMMPC. It is worth noting that there is a slight difference between NTDPC and SPC throughout the entire simulation, highlighting the effectiveness of the designed NTDPC. 
In summary, the figures provide a comprehensive comparison of NTDPC, SPC and SMMPC control strategies, highlighting their performance in terms of output and input trajectories and their respective variability. This comparison aids in understanding the strengths and weaknesses of each method in maintaining desired control objectives amidst system noise and constraints.
\begin{figure}[!ht] 
\centering
\includegraphics[width=300pt]{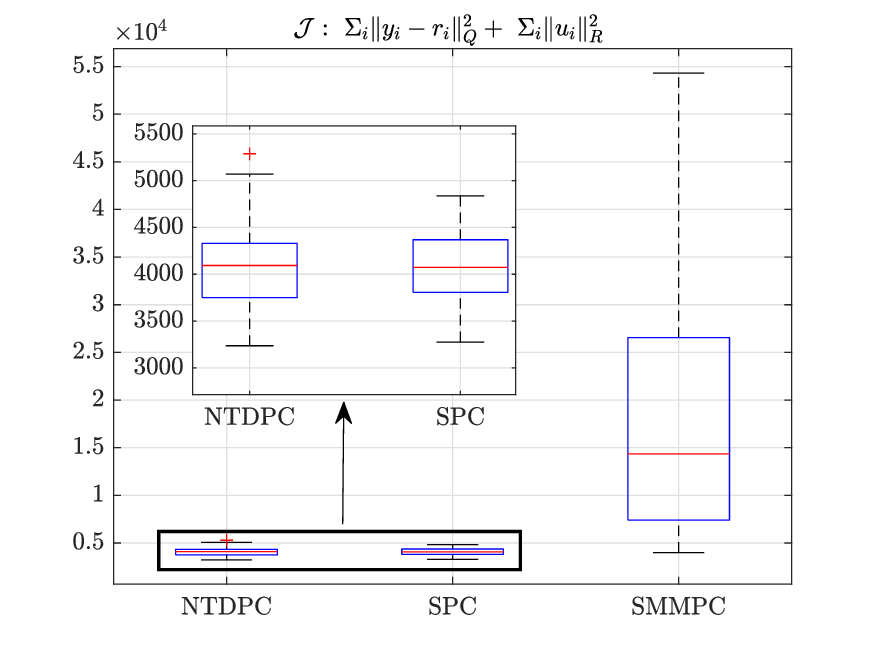}
\caption{Performance index over 50 Monte Carlo runs for noisy case}
 \label{fig:ISE2}
\end{figure}

\subsection{Sensitivity Evaluation}

Increasing horizon lengths improves estimation quality but increases computational cost and adaptation time in time-varying systems. We propose a sensitivity index to guide parameter selection under various noise conditions. When using noisy data, matrix $\Sigma_2$ in \eqref{eq:ss} becomes non-zero. Ideally, $\Sigma_2$'s singular values should be noise-related and significantly smaller than those of $\Sigma_1$. In practice, both matrices contain signal and noise components. Figure \ref{fig:singularvalues} shows singular value patterns for $\Sigma_1$ and $\Sigma_2$ with $\sigma_{\zeta}^2 = 0.01 I$ and $T_{\text{ini}} = N = 15$, with a vertical line separating the two sets.

    \begin{figure}[!ht]
\centering
\includegraphics[width=300pt]{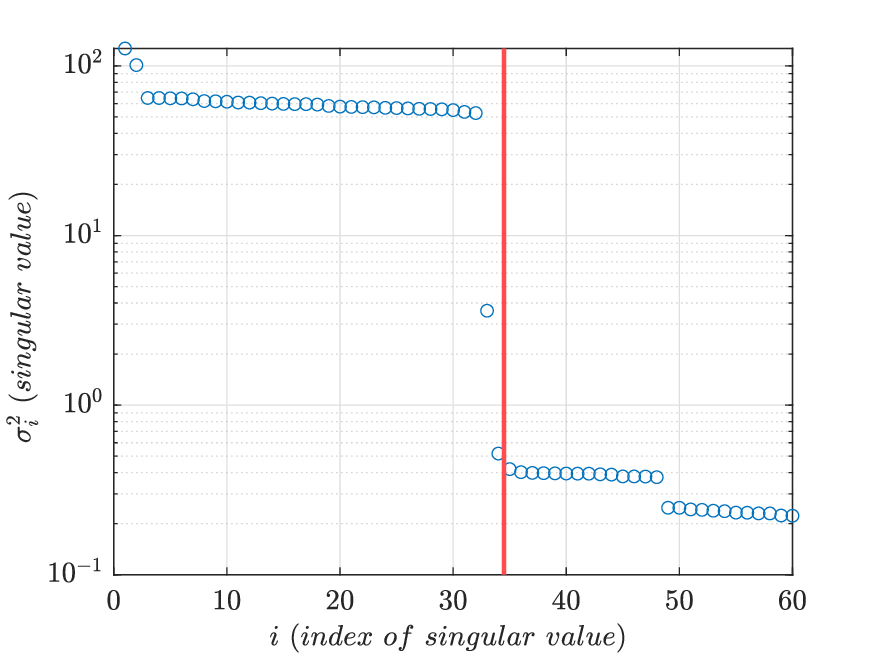}
\caption{Singular value distribution for $T_{\text{ini}}=N=15$ and $\sigma_{\zeta}^2=0.01 I$}
\label{fig:singularvalues}
\end{figure}

To quantify signal-noise separability, we propose:
\begin{equation}
    I_s=\frac{\sigma_{\text{max}}^2(\Sigma_2)}{\sigma_{\text{min}}^2(\Sigma_1)}
    \label{eq:I_s}
\end{equation}
Lower $I_s$ values indicate better signal-noise separability. Figure \ref{fig:ISF} shows $I_s$ behavior as $T_{\text{ini}}$ increases from 10 to 50 with noise variances from 0.01 to 0.32. Empirically, algorithm convergence requires $I_s \lesssim 0.7$. This confirms higher horizons are needed for higher noise variances and guides horizon selection. Each point represents the mean from 50 Monte Carlo simulations, with shading indicating one standard deviation. Since these analyses can be performed before controller implementation (or offline for LTI systems), this approach provides practical guidance for horizon selection.
    \begin{figure}[!ht] 
\centering
\includegraphics[width=300pt]{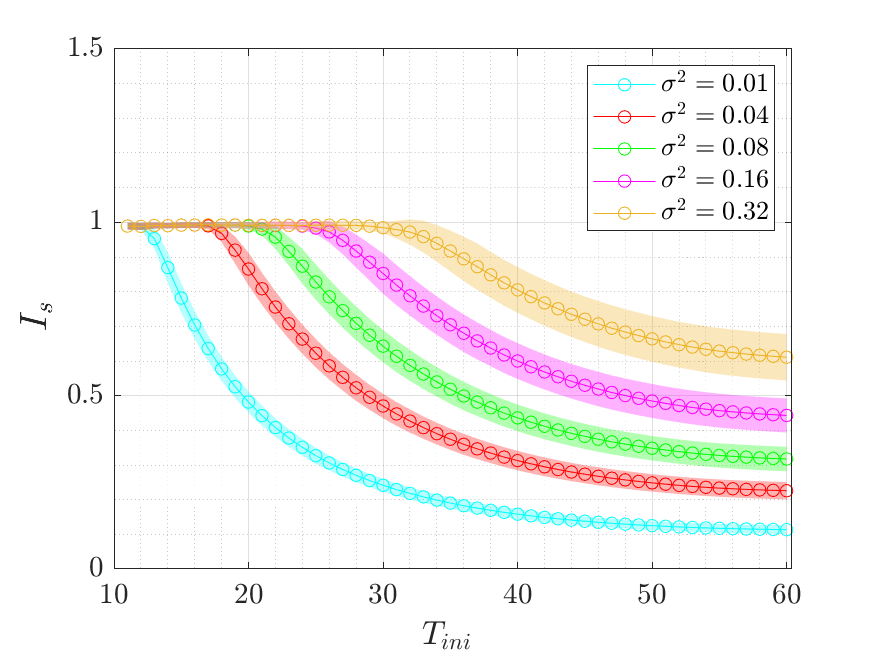}
\caption{Sensitivity index for different noises and $T_{\text{ini}}$ values}
 \label{fig:ISF}
\end{figure}

\subsection{Computational Complexity}
In this section, we will evaluate computational complexity of the developed NTDPC compared to SPC. 
In the context of a time-varying system, the prediction model needs to be updated at each sample time. Comparing the computational effort for updating this model between SPC and the NTDPC, focusing solely on the more costly calculations of pseudoinverse and SVD, reveals the following for a simplified SISO system with input/output horizons $T_{\text{ini}} = N = T_{h}$ (using $T_h$ to avoid confusion with data length $T$) and a data window length $M$:

For SPC, estimating the prediction matrices involves computing the pseudoinverse of the Hankel matrix $\begin{bmatrix} \mathbf{U}_p & \mathbf{Y}_p & \mathbf{U}_f \end{bmatrix}^T$ (or its non-transposed version depending on formulation), which has dimensions $3T_h \times M$. Assuming $M \gg 3T_h$, the computational complexity of the pseudoinverse is approximately $\mathcal{O}((3T_h)^2 M) = \mathcal{O}(9T_h^2 M)$.

For the NTDPC method, the matrix estimation per sample time primarily involves two SVD computations: one for the matrix $Z_p$ (dimensions $2T_h \times M$) and another for the matrix $Z_f V_2$ (dimensions approximately $2T_h \times (M - T_h - n)$, where $n$ is the system order, here simplified to $2T_h \times (M-T_h)$). The computational complexity of the SVDs are approximately $\mathcal{O}((2T_h)^2 M)$ and $\mathcal{O}((2T_h)^2 (M - T_h - n))$, respectively, assuming the number of rows is less than the number of columns in both cases ($2T_h < M$ and $2T_h < M-T_h-n$). Summing these, the total dominant SVD cost for NTDPC is approximately $\mathcal{O}(4T_h^2 M + 4T_h^2 (M - T_h - n)) \approx \mathcal{O}(8T_h^2 M)$ for large $M$.

This comparison of the dominant factorization/inversion costs highlights that NTDPC offers a reduction in computational effort per sample time with a lower constant factor (approximately 8 compared to 9) in the $\mathcal{O}(T_h^2 M)$ complexity. Furthermore, similar to other hybrid methods the NTDPC works with reduced order Hankel matrices which in turn reduces the memory needed for the online matrix estimation step, contributing to improved overall efficiency for time-varying systems.

\section{Conclusion}
This paper presents a comprehensive analysis of the Noise-Tolerant Data-Driven Predictive Control (NTDPC) framework, addressing the critical challenge of measurement noise in hybrid data-driven predictive control while achieving robust tracking of non-zero references. By leveraging singular value decomposition (SVD), NTDPC effectively separates system dynamics from noise within reduced-order Hankel matrices, enabling accurate trajectory predictions with shorter data horizons and reduced computational complexity. The stability analysis demonstrates that NTDPC achieves asymptotic stability in the noise-free case, ensuring precise tracking of the reference equilibrium. In the presence of measurement noise, NTDPC exhibits input-to-state stability (ISS), guaranteeing robust performance with bounded tracking errors proportional to the noise magnitude.

Numerical simulations on a flight control benchmark validate the theoretical findings, showcasing NTDPC’s superior performance compared to SMMPC and SPC. The output trajectories demonstrate stable tracking with damped oscillations and less aggressive control inputs, while the performance index highlights NTDPC’s effectiveness in minimizing tracking errors. The proposed sensitivity index, \( I_s \), provides a practical tool for selecting prediction horizons under varying noise conditions, with empirical results indicating that \( I_s \lesssim 0.7 \) ensures algorithm convergence. This index, computed offline for linear time-invariant systems, enhances the applicability of NTDPC in real-world scenarios. Furthermore, the computational complexity analysis reveals that NTDPC reduces the dominant factorization costs by approximately 11\% compared to SPC, with lower memory requirements due to reduced-order Hankel matrices, making it well-suited for time-varying systems.

\bibliographystyle{unsrt} 
\bibliography{REF}       
\end{document}